\documentclass[english,12pt]{article}
\usepackage[T1]{fontenc}
\usepackage[latin1]{inputenc}
\usepackage{geometry}
\usepackage{float}
\usepackage{mathrsfs}
\usepackage{amsmath}
\usepackage{amssymb}
\usepackage{graphicx}
\usepackage{setspace}
\usepackage{hyperref}
\hypersetup{
    colorlinks=true,
    linkcolor=blue,
    filecolor=magenta,      
    urlcolor=cyan,
    citecolor = red,
}

\makeatletter

\floatstyle{ruled}
\usepackage{algorithm}
\usepackage{algpseudocode}
\usepackage{caption}

\DeclareCaptionFormat{algor}{%
  \hrulefill\par\offinterlineskip\vskip1pt%
    \textbf{#1#2:} #3\offinterlineskip\hrulefill}
\DeclareCaptionStyle{algori}{singlelinecheck=off,format=algor,labelsep=space}

\usepackage{amsthm}

\usepackage{mathrsfs}

\usepackage{amsfonts}

\usepackage{epsfig}

\usepackage{bm}

\usepackage{mathrsfs}

\usepackage{enumerate}

\@ifundefined{definecolor}{\@ifundefined{definecolor}
 {\@ifundefined{definecolor}
 {\usepackage{color}}{}
}{}
}{}

\usepackage{subcaption}

\newtheorem{theorem}{Theorem}[section]

\newtheorem{rem}{Remark}[section]
\newtheorem{prop}{Proposition}[section]

\newtheorem{ass}{Assumption}[section]

\providecommand{\algorithmname}{Algorithm}
\floatname{algorithm}{\protect\algorithmname}

\usepackage[noabbrev, nameinlink]{cleveref}

\crefname{ass}{assumption}{assumptions}
\crefname{prop}{proposition}{propositions}
\crefname{lem}{lemma}{lemmas}
\newcounter{hypA}
\newenvironment{hypA}{\refstepcounter{hypA}\begin{itemize}
  \item[({\bf A\arabic{hypA}})]}{\end{itemize}}
\newcounter{hypB}

\newcounter{hypD}

\usepackage{babel}\date{}

\usepackage{babel}

\makeatother

\usepackage{babel}

\newcommand{\dett}[1]{\text{det}\left[#1\right]}

\usepackage{titlesec}
\titleformat{\paragraph}
{\normalfont\normalsize\bfseries}{\theparagraph}{1em}{}
\titlespacing*{\paragraph}
{0pt}{3.25ex plus 1ex minus .2ex}{1.5ex plus .2ex}


\begin{document}

\begin{center}

\begin{spacing}{1.5}
{ \Large  \textbf{A Lagged Particle Filter for Stable Filtering of certain High-Dimensional State-Space Models}}
\end{spacing}

\vspace{0.4cm}

BY HAMZA RUZAYQAT$^{1}$, AIMAD ER-RAIY$^{1}$, ALEXANDROS BESKOS$^{2}$, DAN CRISAN$^{3}$, AJAY JASRA$^{1}$, \&  NIKOLAS KANTAS$^{3}$  \vspace{0.2cm}

{\footnotesize $^{1}$Computer, Electrical and Mathematical Sciences and Engineering Division, King Abdullah University of Science and Technology, Thuwal, 23955-6900, KSA.}
{\footnotesize E-Mail:\,} \texttt{\emph{\footnotesize hamza.ruzayqat@kaust.edu.sa, aimad.erraiy@kaust.edu.sa, ajay.jasra@kaust.edu.sa}}\\
{\footnotesize $^{2}$Department of Statistical Science, University College London, London, WC1E 6BT, UK.}
{\footnotesize E-Mail:\,} \texttt{\emph{\footnotesize a.beskos@ucl.ac.uk}}\\
{\footnotesize $^{3}$Department of Mathematics, Imperial College London, London, SW7 2AZ, UK.}
{\footnotesize E-Mail:\,} \texttt{\emph{\footnotesize d.crisan@ic.ac.uk, n.kantas@ic.ac.uk}}\\

\end{center}

\begin{abstract}
We consider the problem of high-dimensional filtering of state-space models (SSMs) at discrete times. This problem is particularly challenging as analytical solutions are typically not available and 
many numerical approximation methods can have a cost that scales exponentially with the dimension of the hidden state. Inspired by lag-approximation methods for the smoothing problem \cite{kita,olsson}, 
we introduce a lagged approximation of the smoothing distribution that is necessarily biased. For certain classes of SSMs, particularly those that forget the initial condition exponentially fast in time, 
the bias of our approximation is shown to be uniformly controlled in the dimension and exponentially small in time. 
We develop a sequential Monte Carlo (SMC) method to recursively estimate expectations with respect to our biased filtering distributions.
Moreover, we prove for a class of SSMs that can contain
dependencies amongst coordinates that as the dimension $d\rightarrow\infty$ the cost to achieve a stable mean square error in estimation, for classes of expectations, is of 
$\mathcal{O}(Nd^2)$ per-unit time, where $N$ is the number of simulated samples in the SMC algorithm. Our methodology is implemented on several challenging high-dimensional examples including the conservative shallow-water model.
\\
\\
\noindent \textbf{Keywords}: Filtering, Sequential Monte Carlo, Lag Approximations, High-Dimensional Particle Filter.
\\
\noindent \textbf{Corresponding author}: Hamza Ruzayqat. E-mail:
\href{mailto:hamza.ruzayqat@kaust.edu.sa}{hamza.ruzayqat@kaust.edu.sa} 
\end{abstract}

\section{Introduction}

We are given two sequences of random variables $\{Y_n\}_{n\in\mathbb{N}}$, $\{X_n\}_{n\in\mathbb{N}_0}$, so that $Y_n\in\mathsf{Y}\subseteq\mathbb{R}^{d_y}$, $n\in\mathbb{N}$,  and $X_n\in\mathsf{X}\subseteq\mathbb{R}^d$, $n\in\mathbb{N}_0$. We 
endow $\mathsf{X}$ with 
an associated $\sigma$-field $\mathscr{X}$. 
Consider the state-space model (e.g.~\cite{cappe}), where for $n\in\mathbb{N}$, $(A,B)\subseteq\mathsf{X}\times\mathsf{Y}$:
\begin{align*}
\mathbb{P}(X_n\in A|\{y_p,x_p\}_{p\neq n},y_n) = \int_A f(x_{n-1},x_n) dx_n, \quad \mathbb{P}(Y_n\in B|\{y_p,x_p\}_{p\neq n},x_n) = \int_B g(x_n,y_n) dy_n.
\end{align*}
The initial condition $X_0=x_0\in\mathsf{X}$ is assumed given and fixed, $g(x,\cdot)$ is a positive probability density w.r.t.~the $\sigma$-finite measure $dy$ for each $x\in\mathsf{X}$ and $f(x,\cdot)$ is a positive probability density w.r.t.~the $\sigma$-finite measure $dx$ for each $x\in\mathsf{X}$.
For $n\in\mathbb{N}$ we define the smoothing density:
\begin{align}
\label{eq:smooth}
\pi_n(x_{1:n}) := \frac{\prod_{p=1}^n f(x_{p-1},x_p) g(x_p,y_p)}{\int_{\mathsf{X}^n} \prod_{p=1}^n f(x_{p-1},x_p) g(x_p,y_p) dx_{1:n}}, 
\end{align}
where $y_{1:n}:=(y_1,\ldots,y_n)$ are fixed and known observations.
Our objective is to recursively estimate the so-called filter for each $n\in\mathbb{N}$:
\begin{align*}
\pi_{(n)}(\varphi) := \int_{\mathsf{X}^n}\varphi(x_n)\pi_n(x_{1:n})dx_{1:n},
\end{align*}
where $\varphi:\mathsf{X}\rightarrow\mathbb{R}$ is assumed $\pi_n$-integrable.
This is the filtering problem and has several applications in statistics, applied mathematics and engineering; see for instance \cite{cappe}.

The filtering problem is notoriously challenging for a variety of reasons. The main one is that, with the exception of a small class of models, one cannot
compute the filter analytically. As a result, there is by now a vast literature on the numerical approximation of the filter; see for instance~\cite{cappe,crisan,FK}.
The class of algorithms that we focus upon in this article is based on sequential Monte Carlo (SMC) methods. These techniques generate a collection of $N\geq 1$ samples
in parallel and combine importance sampling and resampling to numerically approximate expectations w.r.t.~the filter. From a mathematical perspective, they are rather well-understood \cite{cappe,FK}, 
with many convergence results as $N$ grows. 

The focus of this article is to study SMC when the dimension of the hidden state, $d$, is large, for example of the order $10^3$ or larger. The high-dimensional filtering problem is often even more problematic than the ordinary filtering problem, i.e. when $d$ is moderate, due to the high costs of numerical implementation. As noted by several authors \cite{chatt,rebs,snyder_et_al}, the importance sampling method
that SMC is based upon can be hugely expensive. The main issue is that for high-dimensional problems of practical interest the proposal and target measure eventually become mutually singular as the 
dimension increases.
As a result, to counteract the weight degeneracy of importance sampling, one may need $N=\mathcal{O}(\kappa^d)$, for some $\kappa>1$ to obtain reasonable estimators. 
This exponential cost in the dimension can be prohibitive in practice and {\em standard} SMC algorithms that consist only of simple importance sampling and resampling recursions 
are not suitable for high-dimensional problems. There do exist more advanced SMC methods which have been successful \cite{beskos2,cotter,kantas,llopis,rebs}, 
but they are often custom designs, taking advantage of useful characteristics of specific problems. 
This article focusses upon further enhancing these type of methods and focuses on a specific class of State Space Models (SSMs),
as opposed to being a universal solution to the high-dimensional filtering problem. 

The starting point of this work is the series of the dimension stability results for SMC samplers \cite{delm:06b} that were obtained in \cite{beskos,beskos1}. These results deal with 
static\footnote{Note that in contrast the filtering problem is inherently dynamic.}, i.i.d.~target sampling problems (i.i.d. here refers to the state coordinates) and show
it is possible to obtain algorithms that scale polynomially in the dimension of the problem and enjoy some type of stability with dimension. 
The results of this type have been utilized in the approaches of \cite{cotter,llopis}, that used a tempering approach to stabilize the weights between two successive observations 
and inserted a SMC sampler in-between data updates of an ordinary 
SMC filtering algorithm (see also \cite{godsill}).
The observation in our article is that whilst empirically such approaches may work well when estimating $\pi_{(n)}$, they do not provide a fully satisfying solution,
because these methods can exhibit the well-known path degeneracy issue for SMC methods.
This is caused by the successive resampling steps and can be briefly described as a lack in diversity in the particles approximating $\pi_n(x_{1:n})$ 
for times $p<n - C N\log N$ with $C>0$ being some model specific constant; see \cite{path_storage,kantas1} for details.
Using tempering adds additional resampling steps and as a result tempering-based methods suitable for high dimensions will suffer from path degeneracy. 
One potential remedy would be to design updates of the entire path, but this is not practical for online algorithms (with fixed computational cost per time). 
However, when this latter strategy is not adopted, we do not believe that it is possible to prove that the algorithm is provably stable as the dimension grows (in fact, we would conjecture that it \emph{collapses})
and this has lead us to our current work.


The contribution of this work is, inspired by the lag-approximation methods for the smoothing problem \cite{doucet06fixedlag,kita,polson_practical,olsson}. We consider a lagged approximation of the smoothing distribution, $\pi_n(x_{1:n})$, 
that is necessarily biased. This approximation induces an independence between the last $L+1$ time points of the hidden states of the SSM and the remaining earlier times.
The basic premise from there is that this new, but biased target, can be numerically approximated with an algorithm that has a cost that scales polynomially with  the dimension. 
These approximations will induce a bias. However, for SSMs that forget their initial condition exponentially fast in time, 
this bias will be uniformly controlled in the dimension and exponentially small in time $L$. In particular, the contributions of this work are:
\begin{itemize}
\item{We propose an appropriate sequence of biased approximations of the smoothers and an SMC algorithm to numerically approximate the related expectations.}
\item{A proof that for SSMs that can contain dependencies amongst coordinates in the transition density and likelihood, as the dimension $d\rightarrow\infty$, the cost to achieve a stable mean square error in estimation 
for certain classes of expectations is of $\mathcal{O}(Nd^2)$ per-unit time.}
\item{Numerical implementation of the SMC algorithm on several challenging and high dimensional examples.}
\end{itemize}
The remaining article is structured as follows. In \autoref{sec:approach} we present our methodology. In \autoref{sec:theory} we present our assumptions and theoretical results.
In \autoref{sec:numerics} the performance of our method is demonstrated on several challenging examples. 
In the Appendix one can find technical results used in the proofs and a detailed description of the algorithm used in \autoref{sec:numerics}.

\paragraph*{Notation} 

Let $(\mathsf{X},\mathscr{X})$ be a measurable space.
$\mathscr{P}(\mathsf{X})$ is the collection of probability measures on $(\mathsf{X},\mathscr{X})$. For $(\mu,\nu)\in\mathscr{P}(\mathsf{X})\times \mathscr{P}(\mathsf{X})$ we write their total variation
distance as $\|\mu-\nu\|_{\mathrm{TV}}$. We use the standard operations 
$(\mu K)(dz):=\int_{\mathsf{X}}\mu(dx)K(x,dz)$, $\mu(\varphi) := \int_{\mathsf{X}}\varphi(x)\mu(dx)$, for $\mu\in \mathscr{P}(\mathsf{X})$, 
transition kernel $K:\mathsf{X}\times \mathscr{X} \to [0,1]$ and $\mu$-integrable $\varphi:\mathsf{X}\to \mathbb{R}$.
$\mathscr{C}(\mathsf{X})$ is the set of real-valued, continuous, measurable functions on $\mathsf{X}$. 
For  bounded $\varphi:\mathsf{X}\to\mathbb{R}$, we set $\|\varphi\|_{\infty}=\sup_{x\in\mathsf{X}}|\varphi(x)|$. For a matrix or vector $A$, we denote by $A^T$ the transpose of $A$.

\section{A Lagged Particle Filter}
\label{sec:approach}
\subsection{Objective}

As has been highlighted in the introduction, the objective is to perform high-dimensional filtering, with a computational cost per-time step that is upper-bounded as the the time parameter (observation time) increases. The computational methods that we shall primarily focus upon are sequential Monte Carlo algorithms such as described in \cite{joh_smc_rev}. These latter methods are often constructed to target the joint smoothing distribution of all the hidden states up-to the given observation time, but as stated e.g.~in \cite{joh_smc_rev}, the approximation that is provided by these algorithms is often only useful to perform filtering rather than smoothing. As a result, we will proceed by first considering the smoothing distribution and a deterministic approximation thereof, but ultimately, we will be using numerical algorithms to estimate expectations associated to an induced deterministic approximation of the filter.

\subsection{The Sequence of Approximate Target Distributions}

Let $L\in\mathbb{N}$ be given and fixed throughout. We assume that one can construct a stochastic method that can approximate reasonably well expectations
w.r.t.~arbitrary probabilities on $\mathsf{X}^{L+1}$, but the method does not work well for higher-dimensional problems. 
By working well we mean that errors do not grow exponentially or worse as the dimension of the hidden state, $d$, increases. We will give more precise details and assumptions related to the SSM and 
SMC components later in \autoref{sec:theory}.

Given this idea, we introduce the sequence of targets approximating the correct smoothing distribution $\pi_{n}(x_{1:n})$ in \eqref{eq:smooth} defined as: 
\begin{align}
\label{eq:target_hat}
&\widehat{\pi}_n(x_{1:n}) \propto 
\alpha_{n-L}(x_{1:n-L})\times \big(\mu_{n-L}(x_{n-L+1}) g(x_{n-L+1},y_{n-L+1})\big)\nonumber\\ &\qquad \qquad \qquad \times 
\Big(\prod_{p=n-L+2}^n f(x_{p-1},x_p) g(x_p,y_p)\Big),\quad n\in\{L+1,L+2,\dots\}.
\end{align}
Here, $\{\mu_{p}\}_{p\ge 1}$ is a sequence of user-specified probability densities on $\mathsf{X}$, giving rise to a corresponding sequence of functions $\{\alpha_{p}\}_{p\ge 1}$, with the latter started by setting 
$\alpha_1(x_1) = g(x_1,y_1) f(x_0,x_1)$, and then obeying the recursion:
\begin{align*}
\alpha_{n-L+1}(x_{1:n-L+1}) = \alpha_{n-L}(x_{1:n-L})\mu_{n-L}(x_{n-L+1})g(x_{n-L+1},y_{n-L+1}), \quad n\in\{L+1,L+2,\dots\}.
\end{align*}

The construction of the sequence of laws $\{\widehat{\pi}_n\}$ is guided by the following principles.
Given index $n\in\{L+1,L+2,\dots\}$, we are interested in the marginal $\widehat{\pi}_{(n)}(x_n)$ being a good approximation of the correct filtering law 
${\pi}_{(n)}(x_n)$. More generally, we focus on the properties of the $L$-dimensional marginal density denoted $\widehat{\pi}_{(n-L+1:n)}(x_{n-L+1:n})$.
Notice that by construction $\widehat{\pi}_n(x_{1:n})$ introduces independence amongst the  random variables $X_{1:n-L}$ and $X_{n-L+1:n}$.
We aim to produce good approximations of the $L$-dimensional marginal, i.e.~of:
\begin{align*}
\widehat{\pi}_{(n-L+1:n)}(x_{n-L+1:n}) &\propto 
\big(\mu_{n-L}(x_{n-L+1})g(x_{n-L+1},y_{n-L+1})\big)\\  &\qquad \qquad \qquad\times  \Big(\prod_{p=n-L+2}^n f(x_{p-1},x_p)g(x_p,y_p)\Big).
\end{align*}
Note that if  $\mu_{n-L}(x_{n-L+1})$ were to be the so-called law $\pi_{(n-L+1|n-L)}(x_{n-L+1})$, that is:
\begin{align}
\label{eq:correct}
\mu_{n-L}(x_{n-L+1}) = \int_{\mathsf{X}}\pi_{(n-L)}(x_{n-L})f(x_{n-L},x_{n-L+1})dx_{n-L} =: \pi_{(n-L+1|n-L)}(x_{n-L+1}),
\end{align}
then $\widehat{\pi}_{(n-L+1:n)}(x_{n-L+1:n}) = \pi_{(n-L+1:n)}(x_{n-L+1:n})$, thus one would be able to perform exact filtering.
However, in our context \eqref{eq:correct} is not of practical use, as even when $n=L+1$, a standard approach for approximating $\int_{\mathsf{X}}\pi_{(n-L)}(x_{n-L})f(x_{n-L},x_{n-L+1})dx_{n-L}$
would typically have a quadratic cost in the number of particles (see \cite{klaas_fast_N2} for a more efficient implementation) and, most importantly, can induce errors growing exponentially in $d$ 
similar to ones seen in \cite{snyder_et_al}.
By choosing a different $\mu_{n-L}(x_{n-L+1})$ we aim for an SMC method that approximates uniformly well in the dimension $d$ expectations of functionals $\varphi(x_{n-L+1:n})$ at the expense of
introducing a bias.
The choice of $\{\alpha_{p}\}_{p\ge 1}$ is motivated directly from the objective of a standard sequential importance sampling scheme when moving from 
$\hat{\pi}_{n-1}$ to $\hat{\pi}_{n}$, where the incremental importance weights will only depend on $x_{n-L:n}$, thereby forgeting earlier positions of $x_{1:n}$. 
%
%



Note that under the choice $\mu_{n-L}(x_{n-L+1})=f(x_{n-L},x_{n-L+1})$, the independence described above will 
typically not hold any more as one reverts back to the original smoothing problem, for which potential Monte-Carlo approximations will degenerate (in general) as the dimension increases.

\subsection{Algorithm}
We  now describe the algorithm that we implement and analyse subsequently. We begin with a generic SMC sampler in
\autoref{alg:smc_samp}. This algorithm, developed in \cite{delm:06b}, allows for the approximation of a sequence of densities of interest, all defined on a 
\emph{common} space.

Our proposed lagged particle filter is detailed in \autoref{alg:lag_filt}. To shorten the description, when $n\in\{1,\dots,L\}$ we set $\widehat{\pi}_n(x_{1:n})=\pi_n(x_{1:n})$.
There are several important remarks at this point. 
In terms of the weight expressions $W_{k+1}^{(i)}$ in \autoref{alg:smc_samp}, when they appear in \autoref{alg:lag_filt} (at iteration $n> L$),
one would have:
\begin{align}
W_{k+1}^{(i)} &= W_{k}^{(i)} \left(\frac{\widehat{\pi}_{n}(x_{1:n}^{(i)})}{\widehat{\pi}_{n-1}(x_{1:n-1}^{(i)})f(x_{n-1}^{(i)},x_{n}^{(i)})}\right)^{\phi_{k+1}-\phi_{k}}  \nonumber \\[0.2cm] 
&=W_{k}^{(i)}\left(\frac{\mu_{n-L}(x_{n-L+1}^{(i)}) g(x_{n}^{(i)},y_{n})}{f(x_{n-L}^{(i)},x_{n-L+1}^{(i)})}\right)^{\phi_{k+1}-\phi_{k}}, \quad \phi_{k+1}-\phi_{k} = \tfrac{1}{d}.
\label{eq:weight_smc} 
\end{align}
Notice, importantly, that this expression does not depend upon $x_{1:n-L-1}^{(i)}$. The choice of spacing the temperatures at order $\mathcal{O}(d^{-1})$ apart coincides with the choice in \cite{beskos}
and will prove critical in the sequel.
A key ingredient of \autoref{alg:lag_filt} is the specification and effectiveness of the Markov kernels $K_{2,n},\dots,K_{d+1,n}$. By construction
the variables $X_{n-L+1:n}$ and $X_{1:n-L}$ (and indeed $X_{n-L:n}$ and $X_{1:n-L-1}$) are independent under the approximate target $\widehat{\pi}_{n}(x_{1:n})$. 
Since the weight function \eqref{eq:weight_smc} depends on $X_{n-L:n}$ only, this means one needs to only update  $X_{n-L:n}$ in the Markov kernel (i.e. the last $(L+1)$
hidden states across time).
\begin{center}
\captionsetup[algorithm]{style=algori}
\captionof{algorithm}{Sequential Monte Carlo Sampler}
\label{alg:smc_samp}
\begin{enumerate}
\item{Input: 
\begin{itemize}
\item{Target density $\kappa(u)$ on state-space $\mathsf{U}$, dominating $\sigma$-finite measure $du$.}
\item{Initial density $\nu(u)$ on state-space $\mathsf{U}$, dominating $\sigma$-finite measure $du$.} 
\item{Annealing parameters $0=\phi_1<\phi_2<\cdots<\phi_p=1$.} 
\item{Sequence of Markov kernels $K_2,\dots,K_p$, such that $K_n$ has invariant distribution proportional to $\kappa(u)^{\phi_n}\nu(u)^{1-\phi_n}du$, $n\in\{2,\dots,p\}$.}
\item{Number of samples $N\in\mathbb{N}$ and resampling threshold $N^*\in[1,N]$.}
\item{Samples $U^{(1)},\dots,U^{(N)}$ that approximate $\nu$.}
\end{itemize}}
\item{Initialize: Set time $k=1$ and $W_k^{(i)}=1$, $u_1^{(i)}=U^{(i)}$ for $i\in\{1,\dots,N\}$.}
\item{Iterate: For $i\in\{1,\dots,N\}$ set:
\begin{align*}
W_{k+1}^{(i)} = W_k^{(i)} \left(\frac{\kappa(u_k^{(i)})}{\nu(u_k^{(i)})}\right)^{\phi_{k+1}-\phi_k}.
\end{align*}
Compute the Effective Sample Size (ESS):
$$
ESS_{k+1} = \frac{\left(\sum_{i=1}^NW_{k+1}^{(i)}\right)^2}{\sum_{i=1}^N(W_{k+1}^{(i)})^2}.
$$
If $ESS_{k+1}<N^*$ resample the particles and set $W_{k+1}^{(i)}=1$ for $i\in\{1,\dots,N\}$, else do nothing.
Then, sample $U_{k+1}^{(i)}|u_k^{(i)}$ using $K_{k+1}(u_k^{(i)},\cdot)$ and set 
$u_{k+1}^{(i)}=U_{k+1}^{(i)}$, $i\in\{1,\dots,N\}$.
 Increment $k=k+1$, if $k=p+1$ go to the output Step 4.~otherwise go to the start of iterate Step 3..}
\item{Output: Samples $U_p^1,\dots,U_p^N$ and corresponding (unnormalised) weights $W_p^1,\dots,W_p^N$.}
\end{enumerate}
\vspace{-0.2cm}
\hrulefill
\vspace{0.3cm}
\end{center}
%
This is typically implemented via a suitable Markov Chain Monte Carlo (MCMC) step such as Metropolis-Hastings or Gibbs sampling. Moreover, as
the Markov kernel can be constructed so that it \emph{does not depend} upon $X_{1:n-L-1}$, the computational cost per iteration of \autoref{alg:lag_filt}
is fixed w.r.t.~the time index $n$.
\begin{center}
\captionsetup[algorithm]{style=algori}
\captionof{algorithm}{A Lagged Particle Filter Algorithm for High-Dimensional Filtering}
\label{alg:lag_filt}
\begin{enumerate}
\item{Initialization: Sample $X_{1}^{(i)}$ i.i.d.~from $f(x_{0},x_1)$. Run the SMC sampler in \autoref{alg:smc_samp} with:
\begin{itemize}
\item{Target density $\pi_1(x_{1})$.} 
\item{Initial density $f(x_{0},x_1)$.} 
\item{Annealing parameters $\phi_k=(k-1)/d$, $p=d+1$, $k\in\{1,2,\ldots,d+1\}$.} 
\item{Sequence of Markov kernels 
$K_{2,1},\dots,K_{d+1,1}$, such that $K_{k,1}$ has invariant distribution proportional to: 
\begin{align}
\label{eq:all_target_1}
\pi_1(x_{1})^{\phi_k}\{f(x_{0},x_1)\}^{1-\phi_k}dx_{1}, \quad k\in\{2,\dots,d+1\}.
\end{align}
}
\item{$N\in\mathbb{N}$, and resampling threshold $N^*=1$ (i.e.~no resampling).}
\item{Initial samples $X_{1}^{(i)}$, $i\in\{1,\dots,N\}$.}
\end{itemize}
Let $X_{1}^{(i)}(1)$, $W^{(i)}(1)$, $i\in\{1,\dots,N\}$, be the samples and weights in the output of \autoref{alg:smc_samp}. Set $n=2$ and go to Step 2.~below.
}
\item{Iteration: Resample $X_{1:n-1}^{(i)}(n-1)$ with weights $W^{(i)}(n-1)$, and denote the obtained particles $\check{X}_{1:n-1}^{(i)}(n-1)$, $i\in\{1,\dots,N\}$. 
Sample new co-ordinate $X_{n}^{(i)}|\check{X}_{n-1}^{(i)}=\check{x}_{n-1}^{(i)}$ using the transition $f(\check{x}_{n-1}^{(i)},\cdot)$, $i\in\{1,\dots,N\}$.
Run the SMC sampler in \autoref{alg:smc_samp} with:
\begin{itemize}
\item{Target density $\widehat{\pi}_{n}(x_{1:n})$.} 
\item{Initial density $\widehat{\pi}_{n-1}(x_{1:n-1})f(x_{n-1},x_{n})$.} 
\item{Annealing parameters $\phi_k=(k-1)/d$, $p=d+1$, $k\in\{1,2,\ldots,d+1\}$.} 
\item{Sequence of Markov kernels 
$K_{2,n},\dots,K_{d+1,n}$, such that $K_{k,n}$ has invariant distribution proportional to: 
\begin{align}
\label{eq:all_target}
\widehat{\pi}_{n}(x_{1:n})^{\phi_k}\{
\widehat{\pi}_{n-1}(x_{1:n-1})f(x_{n-1},x_{n})\}^{1-\phi_k}dx_{1:n},\quad  k\in\{2,\dots,d+1\}.
\end{align}
}
\item{$N\in\mathbb{N}$, $N^*=1$.}
\item{Initial samples $X_{1:n}^{(i)}:=\big(\check{X}_{1:n-1}^{(i)}(n-1),X_{n}^{(i)}\big)$, $i\in\{1,\dots,N\}$.}
\end{itemize}
Let $X_{1:n}^{(i)}(n)$, $W^{(i)}(n)$, $i\in\{1,\dots,N\}$, be the samples, weights in the output of \autoref{alg:smc_samp}. Set $n=n+1$ and return to the start of the iterate at Step 2.~above.}
\end{enumerate}
\vspace{-0.3cm}
\hrulefill
\vspace{0.3cm}
\end{center}
%

In the next section we will prove that, under certain conditions, \autoref{alg:lag_filt} can estimate fixed dimensional marginals with a cost
that is polynomial in dimension. This includes being able to handle the bias that has been introduced by our approximation. Intrinsic to our proof (that is to a large extent
contained in the works of \cite{beskos,beskos1}) is the fact that for large~$d$ (and under certain conditions) the outcome variable $X_{n-L:n}$ resulting from the application of $\mathcal{O}(d)$ tempering steps, will be approximately sampled from the marginal law $\widehat{\pi}_{(n-L:n)}$.
If one does not use the approximation $\{\widehat{\pi}_{n}\}$ and instead works directly with $\{\pi_{n}\}$ then the only way to achieve such an ergodicity would be to
\emph{update the entire path of samples from time 1 to $n$}. This will lead to an algorithm that is not recursive or online. 
%

\section{Theoretical Results}\label{sec:theory}

\subsection{Non I.I.D.~Targets}

To avoid technicalities we henceforth assume that $\mathsf{X}=E^{d}$ for a compact set $E\subset \mathbb{R}$.
We follow closely the approach in \cite{beskos,beskos1}, developed therein for i.i.d.~target distributions.
One can generalize the aforementioned analysis, to the case of models with the following structure.
\begin{ass}
Throughout, $m\in\mathbb{N}_0$ is fixed and $d$ is such that $d\ge m$.
\begin{itemize} 
\label{ass:model}
\item[i)] There exist functions $\overline{g}:E\times\mathsf{Y}\rightarrow\mathbb{R}$ and $\tilde{g}:E^m\times\mathsf{Y}\rightarrow\mathbb{R}$, such that for any $(x,y)\in E^{d}\times\mathsf{Y}$:
$$
g(x,y) = \exp\Big\{\tilde{g}(x^{1:m},y) + \sum_{j=1}^d \overline{g}(x^j,y)\Big\}.
$$

\item[ii)] There exist functions $\overline{f}:E\times E\rightarrow\mathbb{R}$ and $\tilde{f}:E^m\times E^m\rightarrow\mathbb{R}$, such that for any $(x,z)\in E^{2d}$:
$$
f(x,z) = \exp\Big\{\tilde{f}(x^{1:m},z^{1:m}) +\sum_{j=1}^d \overline{f}(x^j,z^j)\Big\}.
$$
\item[iii)] For each $k\in \mathbb{N}$, there exist functions $\overline{\mu}_k:E\rightarrow\mathbb{R}$, $\tilde{\mu}_k:E^m\rightarrow\mathbb{R}$,
 such that for any $x\in E^{d}$:
$$
\mu_k(x) = \exp\Big\{\tilde{\mu}_k(x^{1:m}) + \sum_{j=1}^d \overline{\mu}_k(x^j)\Big\}.
$$
\end{itemize}
%
%
%
\end{ass}
%
%
%
%
\noindent Note we are using here superscripts to index coordinates of a $d$-dimensional vector $x$ and similar to before $x^{p:m}:=(x^p,\ldots,x^m)$, for $m>p$.
Given the model structure in \autoref{ass:model},
the sequence of targets defined via \eqref{eq:all_target_1}-\eqref{eq:all_target}, permits for a factorization 
in terms of co-ordinates $x^{1:m}, x^{m+1}, \ldots, x^{d}$. We will re-express the sequence of targets in a way that exploits this. 

\begin{rem}
The structure in \autoref{ass:model} also implies a factorization for the true target $\pi_n$ as
\begin{align*}
\pi_n(x_{1:n}) 
&\propto \exp\Bigg\{\sum_{j=1}^n \Bigg[ \tilde{f}(x_{j-1}^{1:m},x_j^{1:m})+ \tilde{g}(x_j^{1:m},y_j^{1:m}) + \sum_{k=1}^m \Big(\overline{f}(x_{j-1}^k, x_j^k) + \overline{g}(x_j^k, y_j^k) \Big)\Bigg]\Bigg\} \times\\
&\qquad \qquad \prod_{k=m+1}^d \exp\Bigg\{ \sum_{j=1}^n \Big(\overline{f}(x_{j-1}^k, x_j^k) + \overline{g}(x_j^k, y_j^k) \Big) \Bigg\}.
\end{align*}
\end{rem}

To shorten expressions, we will write for $(x^{1:m},y,x,z^{1:m},z,k)\in E^m\times\mathsf{Y}\times E\times E^m\times E\times\mathbb{N}$:
\begin{align*}
g^{(m)}(x^{1:m},y) &:=  \exp\Big\{\tilde{g}(x^{1:m},y) + \sum_{j=1}^m \overline{g}(x^j,y)\Big\};\\
g^{e}(x,y) &:=  \exp\big\{\overline{g}(x,y)\big\};\\
f^{(m)}(x^{1:m},z^{1:m}) &:=  \exp\Big\{\tilde{f}(x^{1:m},z^{1:m}) +\sum_{j=1}^m \overline{f}(x^j,z^j)\Big\}; \\
f^{e}(x,z) &:=  \exp\big\{\overline{f}(x,z)\big\};\\
\mu_k^{(m)}(x^{1:m}) &:=  \exp\Big\{\tilde{\mu}_k(x^{1:m}) +\sum_{j=1}^m \overline{\mu}_k(x^j)\Big\}; \\
\mu_k^{e}(x) &:=  \exp\big\{\overline{\mu}_k(x)\big\}.
\end{align*}
\begin{rem}
\label{rem:convention}
To avoid repeatitive expressions when separating the initial phase $n\in\{1,2,\dots,L\}$ from subsequent steps $n\in\{L+1,L+2\dots\}$, 
we adopt the convention that distributions when $n<0$ are to be treated as Dirac measures on the fixed position $x_0$ at time $n=0$ (or the subset of co-ordinates of $x_0$ implied by the context).
\end{rem}
%
%
We consider the sequence of tempered approximate target distributions defined via \eqref{eq:all_target_1} for $n=1$ and via \eqref{eq:all_target} for general $n>1$, given the model factorisation introduced by \autoref{ass:model}.
 After simple calculations, we define for $(k,n)\in\{2,\dots,d+1\}\times\mathbb{N}$ and for $j\in\{m+1,\ldots, d\}$ the distributions 
$(\widetilde{\pi}_{k,n}^{(m)},\widetilde{\pi}_{k,n})\in\mathscr{P}(E^{m(L+1)})\times\mathscr{P}(E^{L+1})$ as:
\begin{align}
\widetilde{\pi}_{k,n}^{(m)}(dx_{n-L:n}^{1:m}) 
&\propto \Big(\frac{\mu_{n-L}^{(m)}(x_{n-L+1}^{1:m})g^{(m)}(x_n^{1:m},y_n)}{f^{(m)}(x_{n-L}^{1:m},x_{n-L+1}^{1:m})}\Big)^{\phi_k} \mu^{(m)}_{n-L-1}(x_{n-L}^{1:m})g^{(m)}(x_{n-L}^{1:m},y_{n-L})\nonumber \\ &\qquad \quad \times \Big(\prod_{p=n-L+1}^{n-1}f^{(m)}(x_{p-1}^{1:m},x_p^{1:m})g^{(m)}(x_p^{1:m},y_p)\Big) f^{(m)}(x_{n-1}^{1:m},x_n^{1:m}),\label{eq:pi1} \\[0.4cm]
\widetilde{\pi}_{k,n}(dx_{n-L:n}^{j}) 
&\propto \Big(\frac{\mu_{n-L}^{e}(x_{n-L+1}^j)g^{e}(x_n^{j},y_n)}{f^{e}(x_{n-L}^{j},x_{n-L+1}^{j})}\Big)^{\phi_k} \mu^{e}_{n-L-1}(x_{n-L}^{j})g^{e}(x_{n-L}^{j},y_{n-L})\nonumber\\ &\qquad \quad \times \Big(\prod_{p=n-L+1}^{n-1}f^{e}(x_{p-1}^{j},x_p^j)g^{e}(x_p^j,y_p)\Big) f^{e}(x_{n-1}^j,x_n^j).
\label{eq:pi2}
\end{align}
Given the model structure imposed by \autoref{ass:model},
the mutation kernels used within the algorithm will be adjusted accordingly.
That is,  for $(k,n)\in\{2,\dots,d+1\}\times\mathbb{N}$,
it is reasonable to select  $K_{k,n}:E^{dn}\times\mathcal{B}(E^{dn})\rightarrow[0,1]$
defined as: 
\begin{align*}
K_{k,n}(x_{1:n},dz_{1:n}) = \delta_{\{x_{1:n-L-1}\}}(dz_{1:n-L-1})\times \overline{K}_{k,n}^{(m)}(x_{n-L:n}^{1:m},dz_{n-L:n}^{1:m})\prod_{j=m+1}^d \overline{K}_{k,n}(x_{n-L:n}^j,dz_{n-L:n}^j),
\end{align*}
with
$\overline{K}_{k,n}^{(m)}:E^{m(L+1)}\times\mathcal{B}(E^{m(L+1)})\rightarrow[0,1]$, 
$\overline{K}_{k,n}:E^{(L+1)}\times\mathcal{B}(E^{(L+1)})\rightarrow[0,1]$ 
defined so that they preserve $ \widetilde{\pi}_{k,n}^{(m)}$ and $ \widetilde{\pi}_{k,n}$ respectively, i.e.:
\begin{align*}
 \widetilde{\pi}_{k,n}^{(m)}\overline{K}_{k,n}^{(m)} =  \widetilde{\pi}_{k,n}^{(m)}, \quad 
 \widetilde{\pi}_{k,n} \overline{K}_{k,n} = \widetilde{\pi}_{k,n}.
\end{align*}
%
%
%
%
%
We denote by $x_{k,p}^{(i),j}(n)\in E$  the $i^{\textrm{th}}$ particle ($i\in\{1,\dots,N\}$), in the $j^{\textrm{th}}$ dimension
($j\in\{1,\dots,d\}$),  $k\in\{1,\dots,d\}$ is the SMC sampler time, 
$n$ is the state-space model time and $p$ an instance across times $\{1,\ldots, n\}$.
%
%
Under the model structure determined in \autoref{ass:model}, upon recalling the expression for the weights in \eqref{eq:weight_smc}, one has for $n\in\mathbb{N}$:
%
%
\begin{align*}
&W^{(i)}(n) = \exp\Bigg\{\frac{1}{d}\sum_{k=1}^d \Bigl( \tilde{g}(x_{k,n}^{(i),1:m}(n),y_{n})
+\tilde{\mu}_{n-L}(x_{k,n-L+1}^{(i),1:m}(n))-\tilde{f}(x_{k,n-L}^{(i),1:m}(n),x_{k,n-L+1}^{(i),1:m}(n))\Bigr)  \\[0.2cm]
&\quad \quad + \frac{1}{d}\sum_{j=1}^d\sum_{k=1}^d\Bigl(\overline{g}(x_{k,n}^{(i),j}(n),y_{n})+\overline{\mu}_{n-L}(x_{k,n-L+1}^{(i),j}(n))-\overline{f}(x_{k,n-L}^{(i),j}(n),x_{k,n-L+1}^{(i),j}(n))\Bigr)\Bigg\}.
\end{align*}
In agreement with the convention introduced in \autoref{rem:convention}, 
mappings that involve negative time subscripts should be treated as zeros in the above expression for $W^i(n)$.   
%
%
%

%

\subsection{Error Analysis for Lagged Particle Filter}

To give our theoretical results we introduce some notation following from \cite{beskos}. 
For simplicity, we assume that $x_0^j$ is the same for each $j$.
%
%
%

%
We extend the definition, in \eqref{eq:pi1}-\eqref{eq:pi2}, of  $\widetilde{\pi}_{k,n}^{(m)}$ and $\widetilde{\pi}_{k,n}$,  $n\in\mathbb{N}$,   to allow for a \emph{continuous} tempering parameter $s\in[0,1]$:
\begin{align*}
\widetilde{\pi}_{s,n}^{(m)}(dx_{n-L:n}^{1:m}) 
&\propto \Big(\frac{\mu_{n-L}^{(m)}(x_{n-L+1}^{1:m})g^{(m)}(x_n^{1:m},y_n)}{f^{(m)}(x_{n-L}^{1:m},x_{n-L+1}^{1:m})}\Big)^{s} \mu^{(m)}_{n-L-1}(x_{n-L}^{1:m})g^{(m)}(x_{n-L}^{1:m},y_{n-L})\\ &\qquad \quad \times \Big(\prod_{p=n-L+1}^{n-1}f^{(m)}(x_{p-1}^{1:m},x_p^{1:m})g^{(m)}(x_p^{1:m},y_p)\Big) f^{(m)}(x_{n-1}^{1:m},x_n^{1:m}), \\[0.4cm]
\widetilde{\pi}_{s,n}(dx_{n-L:n}^{j}) 
&\propto \Big(\frac{\mu_{n-L}^{e}(x_{n-L+1}^j)g^{e}(x_n^{j},y_n)}{f^{e}(x_{n-L}^{j},x_{n-L+1}^{j})}\Big)^{s} \mu^{e}_{n-L-1}(x_{n-L}^{j})g^{e}(x_{n-L}^{j},y_{n-L})\\ &\qquad \quad \times \Big(\prod_{p=n-L+1}^{n-1}f^{e}(x_{p-1}^{j},x_p^j)g^{e}(x_p^j,y_p)\Big) f^{e}(x_{n-1}^j,x_n^j).
\end{align*}
Similarly, for $s\in[0,1]$, $n\in\mathbb{N}$, we consider the \emph{continuum} of Markov kernels $\overline{K}_{s,n}^{(m)}$ and $\overline{K}_{s,n}$ that preserve $\widetilde{\pi}_{s,n}^{(m)}$ and $\widetilde{\pi}_{s,n}$, respectively. 

For $n\in\mathbb{N}$ and $x_{n-L:n}\in E^{L+1}$ set:
\begin{align*}
\psi_n(x_{n-L:n}) := \overline{g}(x_{n},y_{n})+\overline{\mu}_{n-L}(x_{n-L+1})-\overline{f}(x_{n-L},x_{n-L+1}).
\end{align*}
Standard adjustments apply in the case of negative time index, see \autoref{rem:convention}.
We denote by $\hat{\psi}_{s,n}:E^{L+1}\to \mathbb{R}$ the solution to the following Poisson equation, $s\in[0,1]$, $n\in \mathbb{N}$:

\begin{align*}
 \hat{\psi}_{s,n}-\overline{K}_{s,n}(\hat{\psi}_{s,n}) = \psi_n-\widetilde{\pi}_{s,n}(\psi_n). 
\end{align*} 
We set 
\begin{align*}
\sigma^2(n) = \int_{0}^1\widetilde{\pi}_{s,n}(\hat{\psi}_{s,n}^2-\overline{K}_{s,n}(\hat{\psi}_{s,n})^2)ds.
\end{align*} 

We use the following assumptions as in \cite{beskos1}. \vspace{0.2cm}

\begin{hypA}\label{hyp_a:1}
For any $n\in\{1,2,\ldots,\}$ there exist $\theta_n\in(0,1)$,  $(\zeta_n^{(m)},\zeta_n)\in\mathscr{P}(\mathcal{B}(E^{m(L+1)}))\times
\mathscr{P}(\mathcal{B}(E^{L+1}))$, such that for each $s\in[0,1]$:
%
%
%
%
\begin{align*}
\overline{K}_{s,n}^{(m)}(\overline{x}_{n-L:n}^{1:m},A) & \geq  \theta_n \zeta_n^{(m)}(A),\quad  \textrm{for all }(x_{n-L:n}^{1:m},A)\in E^{m(L+1)}\times\mathcal{B}(E^{m(L+1)}), \\[0.2cm]
\overline{K}_{s,n}(\overline{x}_{n-L:n},A) & \geq  \theta_n \zeta_n(A),\quad \textrm{for all } (\overline{x}_{1:L+1},A)\in E^{L+1}\times\mathcal{B}(E^{L+1}).
\end{align*}
\end{hypA}

\begin{hypA}\label{hyp_a:2}
For any $n\in\{1,2,\ldots,\}$ there exists $C_n<\infty$ such that
for any $s,t\in[0,1]$ we have
\begin{align*}
\sup_{x_{n-L:n}^{1:m}\in E^{m(L+1)}}\big\|\overline{K}_{s,n}^{(m)}(x_{n-L:n}^{1:m},\cdot)-\overline{K}_{t,n}^{(m)}(x_{n-L:n}^{1:m},\cdot)\big\|_{\mathrm{TV}} & \leq  C_n|s-t|,\\
\sup_{x_{n-L:n}\in E^{L+1}}\big\|\overline{K}_{s,n}(x_{n-L:n},\cdot)-\overline{K}_{t,n}(x_{n-L:n},\cdot)\big\|_{\mathrm{TV}} & \leq  C_n|s-t|.
\end{align*}
\end{hypA}

\begin{hypA}\label{hyp_a:3}
There exists a $C<\infty$ such that:
\begin{itemize}
\item[i)] $\max\{\sup_{(x^{1:m},y)\in E^m\times\mathsf{Y}}|\tilde{g}(x^{1:m},y)|,\sup_{(x,y)\in E\times\mathsf{Y}}
|\overline{g}(x,y)|\}\leq C$;
\item[ii)] $\max\{\sup_{(x^{1:m},z^{1:m})\in E^{2m}}|\tilde{f}(x^{1:m},z^{1:m})|,\sup_{(x,z)\in E^2}|\overline{f}(x,z)|\}\leq C$;
\item[iii)] $\max\{\sup_{(x^{1:m},k)\in E^m\times\mathbb{N}}|\tilde{\mu}_k(x^{1:m})|,\sup_{(x,k)\in E\times\mathbb{N}}|\overline{\mu}_k(x)|\}\leq C$. \vspace{0.2cm}
\end{itemize}
\end{hypA}

\begin{rem}
To simplify the notation, for $s\in[0,1]$, $n\in\mathbb{N}$, and $n-L\leq k \leq n$,
we denote by $\widetilde{\pi}_{s,n,k}^{(m)}$ 
and $\widetilde{\pi}_{s,n,k}$ the marginal of laws  $\widetilde{\pi}_{s,n}^{(m)}(dx^{1:m}_{n-L:n})$
and $\widetilde{\pi}_{s,n}(dx_{n-L:n})$ on $x_{k:n}^{1:m}$ and $x_{k:n}$, respectively.
In addition, $\mathbb{V}_{\widetilde{\pi}_{s,n,k}^{(m)}}$, $\mathbb{V}_{\widetilde{\pi}_{s,n,k}}$ denote the variances of these above laws.
%
%
%
%
\end{rem}

Now we have the following result, that constitutes an extension of \cite[Theorem 3.1]{beskos1}. 
Here, $\|\cdot\|_2$ denotes the standard $L_2$-norm on the space of squared-integrable random variables.

\begin{prop}\label{prop:1}
Assume \autoref{ass:model} and (\textbf{A}\ref{hyp_a:1}-\ref{hyp_a:3}) hold. Then for any fixed $n$ and number of particles $N$ with $(n,N)\in\mathbb{N}\times\mathbb{N}$, we have:
%
%
%

%

\begin{enumerate}
\item[i)] For any $\varphi^{(m)}\in\mathscr{C}(E^{m(L+1)})$ 
\begin{align*}
&\lim_{d\rightarrow\infty}\Big\| \tfrac{\sum_{i=1}^N \varphi^{(m)}(\check{X}_{n-L:n}^{(i),1:m}(n))}{N}-\widetilde{\pi}_{1,n,n-L}^{(m)}(\varphi^{(m)})\Big\|_2^{2} \leq 
\tfrac{\mathbb{V}_{\widetilde{\pi}_{1,n,n-L}^{(m)} }[\varphi^{(m)}]}{N}\big(1+C_{n,N}(\sigma^2(n))\big).
\end{align*}

\item[ii)] For any $\varphi\in\mathscr{C}(E^{L+1})$
\begin{align*}
&\lim_{d\rightarrow\infty}\Big\|\tfrac{\sum_{i=1}^N\varphi(\check{X}_{n-L:n}^{(i),m+1}(n))}{N}-\widetilde{\pi}_{1,n,n-L}(\varphi)\Big\|_2^2  \leq
\tfrac{\mathbb{V}_{\widetilde{\pi}_{1,n,n-L} }[\varphi]}{N}\big(1+C_{n,N}(\sigma^2(n))\big)
\end{align*}

\end{enumerate}

\noindent In both cases above  $C_{n,N}(\sigma^2(n)) = \exp\{\sigma^2(n)\} + C_n \exp\{17\sigma^2(n)\}/N^{1/6}$, where $C_n$ is a constant that does not depend on $N$ or $\sigma^2(n)$.


\end{prop}

\begin{proof}
The proof is given in Appendix \ref{sec:proofP}.
\end{proof}

\subsection{Error Analysis for Lagged Distribution Approximation}

We add an assumption which will be useful below. \vspace{0.2cm}

\begin{hypA}\label{hyp_a:4}
We have the following lower bounds, for some $0<C<\infty$:
\begin{itemize}
\item[i)] $\min\{\inf_{(x,y)\in E^m\times\mathsf{Y}}\tilde{g}(x^{1:m},y),\inf_{(x,y)\in E\times\mathsf{Y}}\overline{g}(x,y)\}\geq C$.
\item[ii)] $\min\{\inf_{(x^{1:m},z^{1:m})\in E}\tilde{f}(x^{1:m},z^{1:m}),\inf_{(x,z)\in E^2}\overline{f}(x,z)\}\geq C$. \vspace{0.2cm}
\end{itemize}
\end{hypA}

We write $\widetilde{\pi}_{(n)}^{(m)}$ (resp.~$\widetilde{\pi}_{(n-L+1|n-L)}^{(m)}$) and $\widetilde{\pi}_{(n)}$ (resp.~$\widetilde{\pi}_{(n-L+1|n-L)}^{(m)}$) 
for the \emph{true} filtering distribution (resp.~\emph{true} predictive distribution as defined in \eqref{eq:correct}) in the first $m$ dimensions and in the $(m+1)^{\textrm{th}}$-dimension respectively. 
Then the following result is standard in the literature, see e.g.~\cite[Proposition 4.3.7]{FK}. 

\begin{prop}\label{prop:2}
Assume (\textbf{A}\autoref{hyp_a:3}-\autoref{hyp_a:4}). Then there exist $(\rho,C)\in(0,1)\times(0,\infty)$ such that: 
\begin{enumerate}
\item[i)] For
any $(d,L,n)\in\mathbb{N}^3$ and any bounded function $\varphi^{(m)}:E^{m}\to \mathbb{R}$: 
\begin{align*}
\Big|\,\big[\widetilde{\pi}_{1,n,n}^{(m)}-\widetilde{\pi}_{(n)}^{(m)}\big](\varphi^{(m)})\,\Big| \leq C\,\big\|\varphi^{(m)}\big\|_{\infty} \big\|\mu_{n-L}^{(m)}-\widetilde{\pi}_{(n-L+1|n-L)}^{(m)}\big\|_{\mathrm{TV}}\times \rho^L.
\end{align*}
 
\item[ii)] For
any $(d,L,n)\in\mathbb{N}^3$ and any bounded function $\varphi:E\to \mathbb{R}$: 
\begin{align*}
\Big|\,\big[\widetilde{\pi}_{1,n,n}-\widetilde{\pi}_{(n)}\big](\varphi)\,\Big| \leq C\,\big\|\varphi\big\|_{\infty} \big\|\mu_{n-L}^e-\widetilde{\pi}_{(n-L+1|n-L)}\big\|_{\mathrm{TV}}\times \rho^L.
\end{align*}
 
\end{enumerate}
\end{prop}

\begin{rem}
Condition (\textbf{A}\autoref{hyp_a:3})(i) can be relaxed to the following: for any $y\in\mathsf{Y}$, there exists a $C<\infty$ such that $\max\{\sup_{x^{1:m}\in E^m}\tilde{g}(x,y),\sup_{x\in E}\overline{g}(x,y)\}\leq C$.
This weaker condition suffices for \autoref{prop:1}, but not for \autoref{prop:2}.
\end{rem}

\subsection{Total Error of Lagged Approach}

As a result of Propositions \ref{prop:1}-\ref{prop:2} it easily follows that for $n\in\mathbb{N}_0$ 
\begin{align*}
&\lim_{d\rightarrow\infty}\Big\|\sum_{i=1}^N\tfrac{\varphi^{(m)}(\check{X}_{n}^{(i),1:m}(n))}{N}-\widetilde{\pi}_{(n)}^{(m)}(\varphi^{(m)})\Big\|_{2}^{2} \leq \\
&\qquad \qquad \tfrac{\mathbb{V}_{ \widetilde{\pi}_{1,n,n}^{(m)} }[\varphi^{(m)}]}{N}\Big(1+C_{n,N}(\sigma^2(n))\Big) + 
C\,\big\|\varphi^{(m)}\big\|_{\infty}^2 \big\|\mu_{n-L}^{(m)}-\widetilde{\pi}_{(n-L+1|n-L)}^{(m)}\big\|^2_{\mathrm{TV}}\times \rho^L,
\end{align*}
and
\begin{align*}
&
\lim_{d\rightarrow\infty}
\Big\| 
\tfrac{\sum_{i=1}^N\varphi(\check{X}_{n}^{(i),m+1}(n))}{N}-\widetilde{\pi}_{(n)}(\varphi)
\Big\|_2^2 \leq \\ 
\\
& \qquad \qquad\tfrac{\mathbb{V}_{\widetilde{\pi}_{1,n,n-L}}[\varphi]}{N}\Big(1+C_{n,N}(\sigma^2(n))\Big) + 
C\,\big\|\varphi\big\|_{\infty}^2 \big\|\mu_{n-L}^{e}-\pi_{(n-L+1|n-L)}\big\|^2_{\mathrm{TV}}\times \rho^L.
\end{align*}
Thus, for given $\epsilon>0$, to obtain an MSE of $\mathcal{O}(\epsilon^2)$ one needs to choose $L$ large enough, \emph{independently of $d$} and then a number of particles large enough, \emph{independently of $d$}. In particular, choosing $L=\mathcal{O}(|\log(\epsilon)|)$ and $N=\mathcal{O}(\epsilon^{-2})$ given an  MSE is $\mathcal{O}(\epsilon^2)$.
The cost to achieve this MSE, with $n>L$, is of the order $\mathcal{O}(d^2\epsilon^{-2}n|\log(\epsilon)|)$. Due to the structure of the model, we expect that this cost is a best case scenario.
Note also, that one cannot in practice increase $L$ indefinitely, as typically the mixing rate of the Markov kernels $\{K_{k,n}\}_{(k,n)\in\{2,\dots,d+1\}\times\{1,2,\dots\}}$ would fall  with $L$; also, one cannot afford to run the algorithm with $L$ too large. To make the latter mixing rate independent of $L$, one would have to iterate the kernels, for instance, in the case of Metropolis-Hastings kernels with random walk proposals, this would likely have to be at least $\mathcal{O}(L)$ steps (see \cite{rob_opt_scale}), leading to a computational cost of $\mathcal{O}(d^2\epsilon^{-2}n|\log(\epsilon)|^2)$.

\section{Numerical Results}\label{sec:numerics}

\subsection{Model}
Let $X_n \in \mathbb{R}^d$, for $n=1,\ldots,T$, and $Y_m\in \mathbb{R}^{d_y}$, for $m=1,\ldots,M$, for some $T,M \in \mathbb{N}$. Assume that $X_0=x_0\in \mathbb{R}^d$ is given. We consider the following SSM: 
\begin{align}
X_n &= q_n(X_{n-1}) + R_1^{1/2} W_n, \qquad n = 1, \ldots,T, \label{eq:signal}\\
Y_m &= C X_{m\hat{k}} + R_2^{1/2} V_m, \qquad m = 1,\ldots,M,\label{eq:obs}
\end{align} 
\begin{figure}[h!]
\begin{subfigure}[c]{0.48\textwidth}
\includegraphics[width = \textwidth]{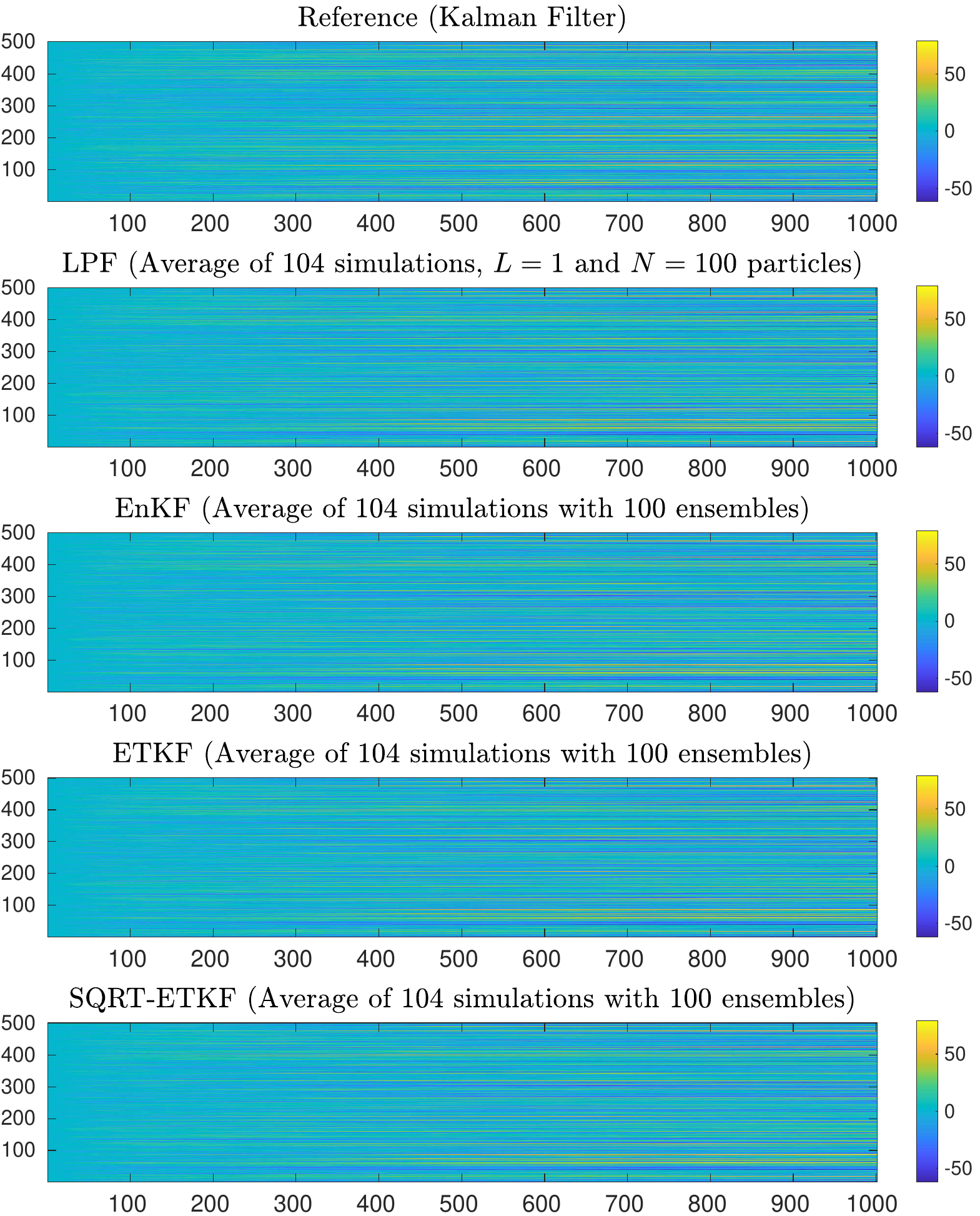}
\end{subfigure}
\begin{subfigure}[c]{0.51\textwidth}
\includegraphics[width =\textwidth]{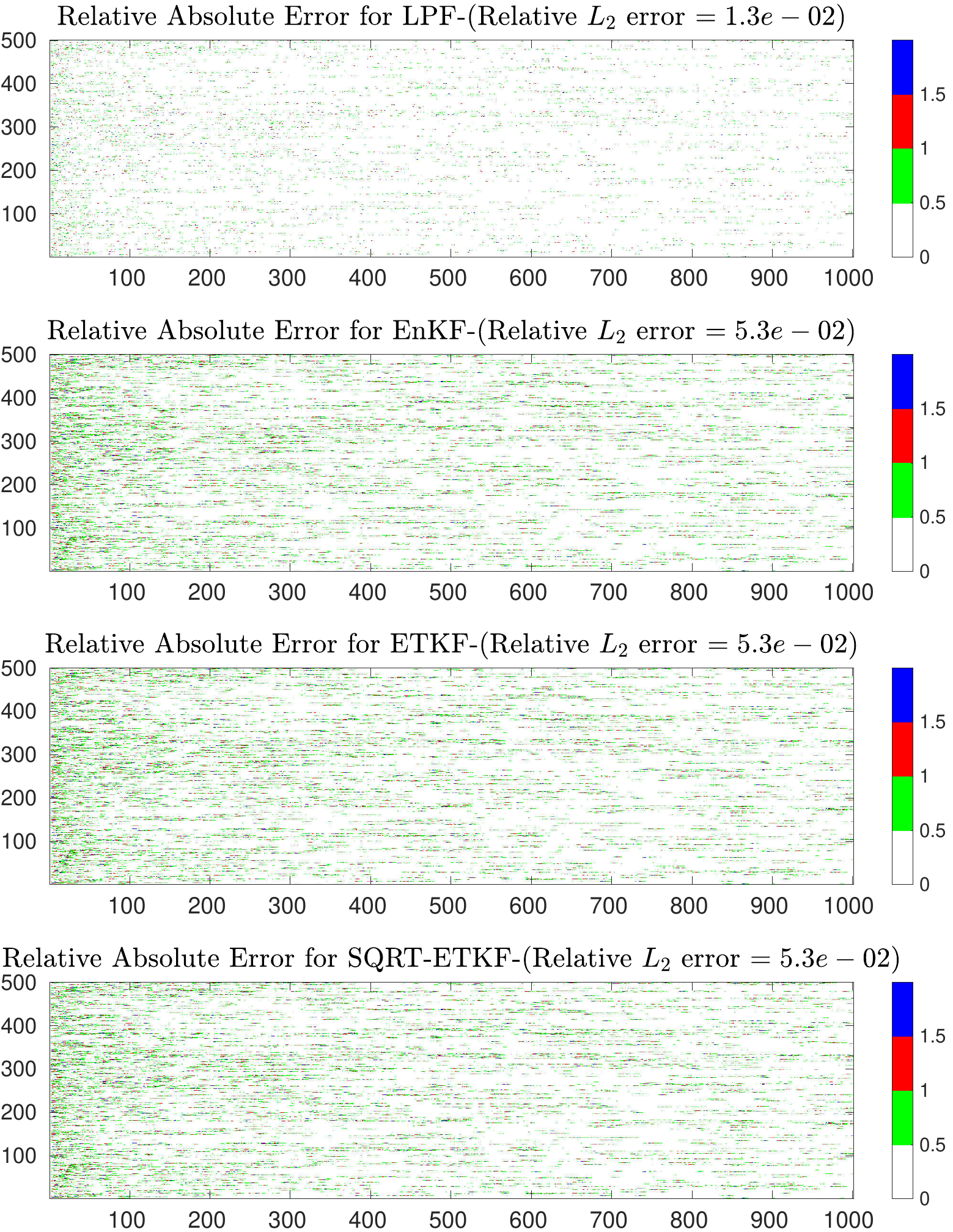}
\end{subfigure}
\caption{(Linear-Gaussian Model) A comparison amongst the different filters. The horizontal axis represents the time parameter $n$, and the vertical axis represents the signal's coordinates. 
The left panel shows approximations of $\pi_n(\varphi)$ with $\varphi(x_{1:n})=x_{n}$ as obtained by the KF, LPF, EnKF, ETKF and ETKF-SQRT algorithms. The right panel shows the corresponding relative absolute errors and the corresponding relative $L_2$-error is shown in the titles. For readability, the color of the relative absolute errors that are less than 0.5 is set to white.}
\label{fig:Linear_mean_errors}
\end{figure}
\begin{figure}[h!]
\centering
\includegraphics[width =0.7\textwidth, height = 0.5\textheight]{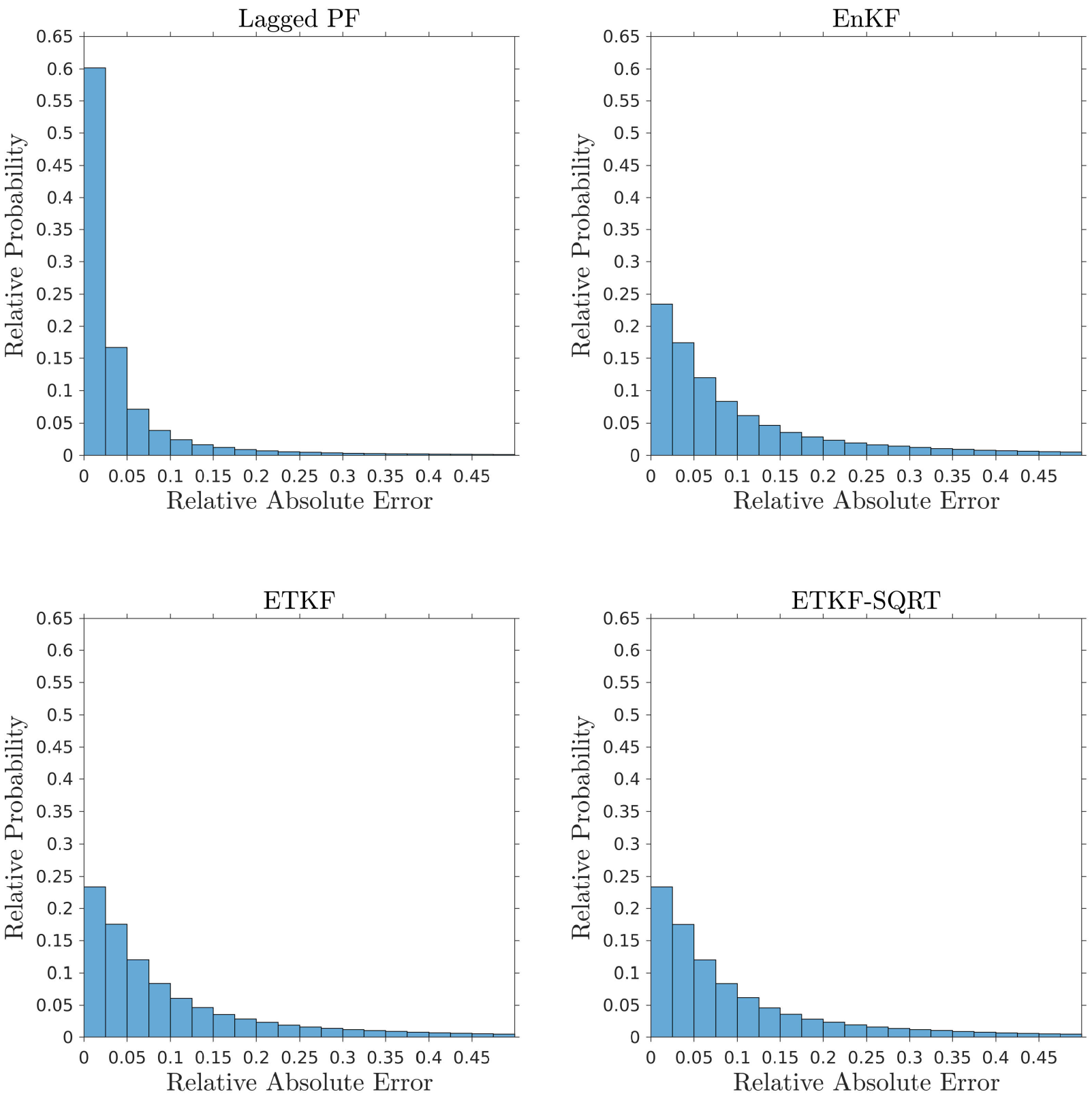}
\caption{(Linear-Gaussian Model) A histogram of the relative absolute errors for each filter. The relative probability here is defined as the number of elements in the bin divided by the total number of elements, which is $d \times (T+1)$.}
\label{fig:Linear_histo}
\end{figure}
\begin{figure}[h!]
\centering
\includegraphics[width =0.8\textwidth]{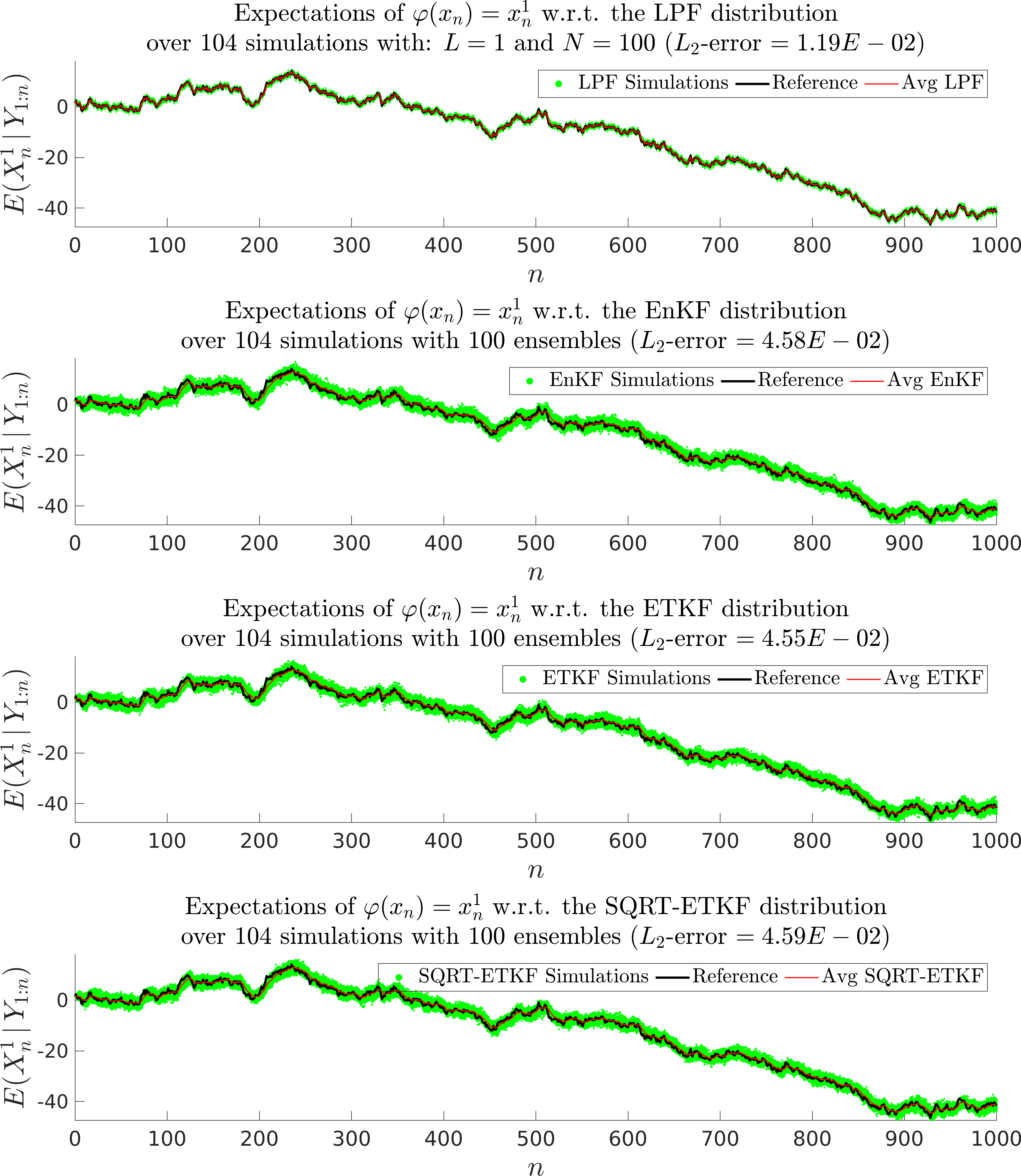}
\caption{(Linear-Gaussian Model) A comparison of approximations of $\pi_n(\varphi)$ with $\varphi(x_{1:n})=x_n^{1}$ as found by~LPF, EnKF, ETKF and ETKF-SQRT. The green circle/dots represent the cloud of different simulations of each filter. The black curve is the reference from the KF and the red curve is the average over 104 independent simulations of each method.}
\label{fig:Linear_1coord}
\end{figure}
\begin{figure}[h!]
\centering
\begin{subfigure}[c]{0.55\textwidth}
\includegraphics[width =\textwidth]{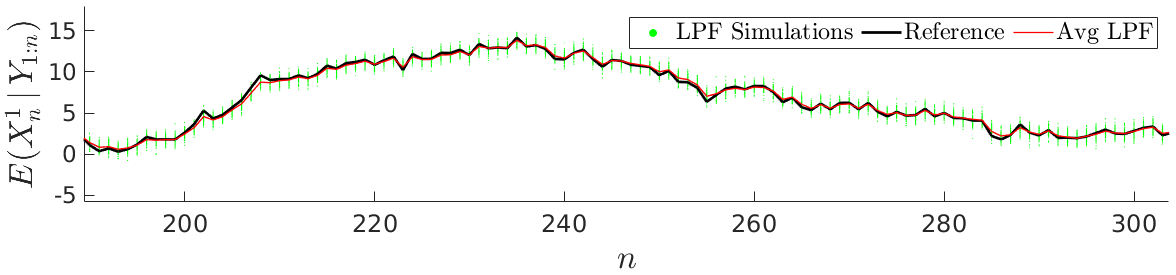}
\end{subfigure}\\
\begin{subfigure}[c]{0.55\textwidth}
\includegraphics[width =\textwidth]{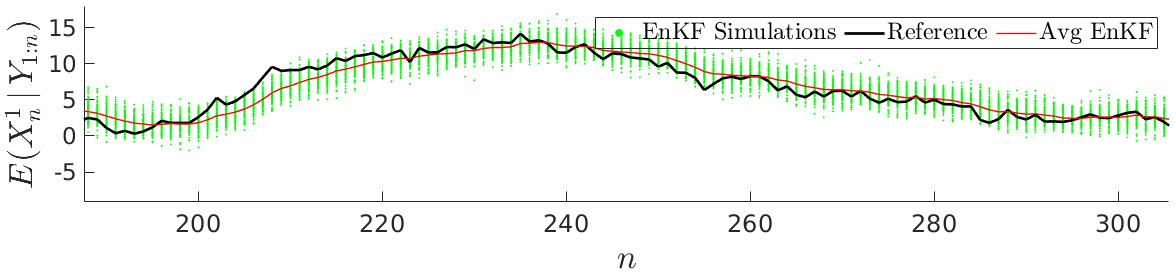}
\end{subfigure}
\caption{(Linear-Gaussian Model) We zoom in the region $[190,305]\times [-5,15]$ in \autoref{fig:Linear_1coord} for the LPF and EnKF (the rest of the methods are very similar to the EnKF.)}
\label{fig:Linear_1coord_zoom}
\end{figure}
where $R_1 \in \mathbb{R}^{d \times d}$, $R_2 \in \mathbb{R}^{d_y \times d_y}$ are symmetric positive definite matrices, and $C \in \mathbb{R}^{d_y \times d}$. The random variables $W_n$ and $V_m$ are sampled from $ \mathcal{N}(0,I_{d})$ and $\mathcal{N}(0,I_{d_y})$, respectively, where $I_r$ is the identity matrix of size $r \times r$. The positive integer $\hat{k}$ is the time frequency at which the signal or part of it is observed (e.g. if $\hat{k}=3$, the signal is observed at times $3,6,9,\ldots$.) The function $q_n:\mathbb{R}^d \to \mathbb{R}^d$ is possibly nonlinear. 

The target smoothing distribution is
\begin{align}
\label{eq:target}
\pi_n(x_{1:n})\propto \prod_{i=1}^n f(x_{i-1},x_i)\,g(x_i,y_i), 
\end{align}
where from the above SSM,
\begin{align}
\label{eq:Model}
f(x_{n-1},x_n) &= \frac{1}{\sqrt{(2\pi)^{d}\dett{R_1}} } \exp\left\{-\tfrac{1}{2}[x_n-q_n(x_{n-1})]^\top R_1^{-1} [x_n-q(x_{n-1})] \right\}, \\
g(x_{m\hat{k}},y_m) &= \frac{1}{\sqrt{(2\pi)^{d_y}\dett{R_2}} } \exp\left\{-\tfrac{1}{2}[y_m-Cx_{m\hat{k}}]^\top R_2^{-1} [y_m-Cx_{m\hat{k}}] \right\}.
\end{align}

We will implement an enhancement of \autoref{alg:lag_filt} that includes adaptive tempering and resampling. Resampling occurs whenever the ESS $\mathcal{E}$ is less than a given threshold $N^*$. 
This way the algorithm calculates the annealing parameters $\phi_k$ on the fly. A detailed description of the algorithm used can be found
in \autoref{alg:LPF} in the Appendix. There we also specify the Markov kernels $K_{k,n}$ as Gaussian random walk Metropolis updates.
The proposal density covariance matrix $\Sigma_m$ is proportional to $(2.38^2/d) \alpha(\phi_k)I_d$ modified adaptively 
so that the average acceptance rate over all particles is in the range 0.15 - 0.25. Here $\alpha(\phi_k)$ is some function of the annealing parameter $\phi_k$, e.g. $\alpha(\phi_k)=(\phi_k+2)/(\phi_k+1)$. We use an $S$ iterations of $K_{k,n}$, with $S$ varying between $15$ and $25$. 

In the next subsection, we test the algorithm on the above model with three choices of the function $q_n$ and three corresponding different scenarios for the observations. In the first example, the signal is fully observed and the time frequency is $\hat{k} = 1$. For the second example, the signal is fully observed and $\hat{k}=3$. In the third example, we observe around 40\% of the state coordinates and the time frequency is $\hat{k} = 1$. Then, we compare results from the lagged particle filter (LPF) \autoref{alg:LPF}, the ensemble Kalman filter (EnKF) \cite{even94,even03}, the ensemble transform Kalman filter (ETKF) \cite{bem} and the ETKF with square-root (ETKF-SQRT) \cite{njsh,so}.

\subsubsection{Linear-Gaussian Model}
\label{subsubsec:Linear Model}
Here, we set $q_n(x) = x$, $T=10^3$, $N = 100$, $N^* = 0.8~N$, $d=d_y=500$, $R_1^{1/2}=\frac{1}{\sqrt{2}}I_d$, $R_2^{1/2}=0.1I_{d_y}$, $C=I_d$, $X_0=1.5\cdot\textbf{1}$, where $\textbf{1}$ is a vector of ones. We consider a lag of 2, i.e., $L=1$. For $n\geq L+1$, we take $\mu_{n-L}(x_{n-L+1})=\mathcal{N}(x_{n-L+1}; m_{n-L+1},P_{n-L+1})$, where $m_{n-L+1}$ and $P_{n-L+1}$ are the mean and the covariance of the Kalman predictor distribution at time step $n-L+1$. The time frequency of the observations is $\hat{k}=1$, and thus $M=10^3$.

In the left panel of \autoref{fig:Linear_mean_errors}, we plot the average of $\pi_\cdot(\varphi)$ with $\varphi(x_{1:n}) = x_n$ after 104 independent runs. 
Here the subscript $\cdot$ refers to the filter of interest (LPF, EnKF, ETKF, ETKF-SQRT). Since this model is linear-Gaussian, we use the Kalman filter (KF) as a reference. 
The right panel of \autoref{fig:Linear_mean_errors} shows the relative absolute error for each filter.
Additionally, the relative $L_2$-error $\|\pi_\cdot(\varphi) - \pi_{KF}(\varphi)\|_2 / \| \pi_{KF}(\varphi)\|_2$ is shown at the title of each subfigure. 
We elaborate further in \autoref{fig:Linear_histo} and present a histogram of the relative absolute errors. The histogram shows that for the same number of particles/ensembles, 
around 60\% of the LPF errors are less than 0.025 and only 23\% in the rest of the methods. Finally, \autoref{fig:Linear_1coord} shows a cloud plot of estimates of $\mathbb{E}(x_n^1|y_{1:n})$ for each independent run over time, where $1\leq n \leq T$ and $x_n^{1}$ being the first coordinate of $x_n$. The LPF method tracks the KF very well on average (as illustrated in \autoref{fig:Linear_1coord_zoom}) and it shows considerable less variability between different independent runs.
Similar comparisons will be carried for the other examples to follow.

\subsubsection{Lorenz 96 Model}
We test the algorithm on the Lorenz 96 model \cite{L96}, perturbed by a Gaussian noise, along with noisy observations.
\begin{figure}[h!]
\centering
\begin{subfigure}[c]{0.49\textwidth}
\includegraphics[width = \textwidth, height= 0.43\textheight]{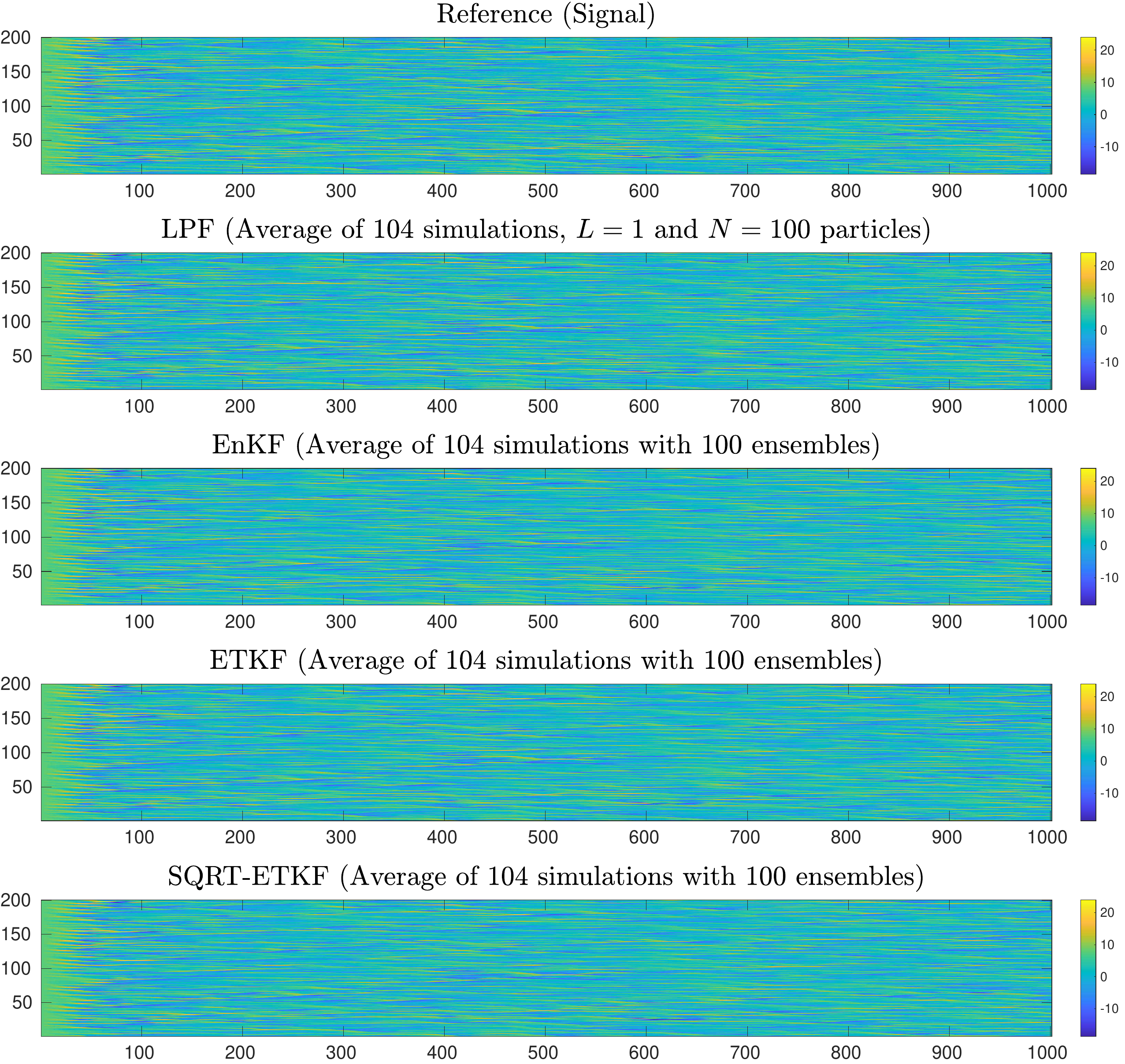}
\end{subfigure}
\begin{subfigure}[c]{0.49\textwidth}
\includegraphics[width =0.99\textwidth, height =0.43\textheight]{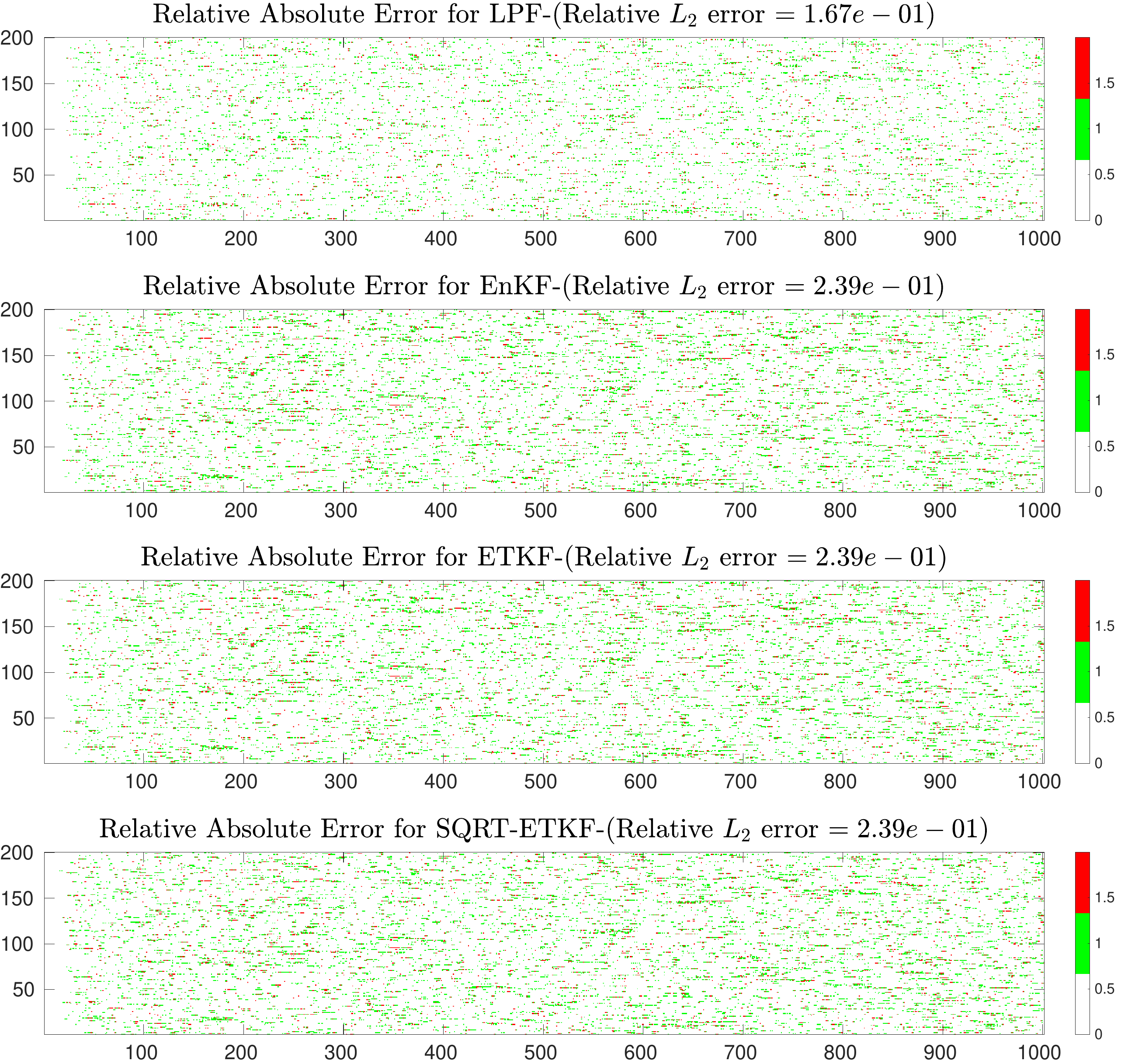}
\end{subfigure}
\caption{(Lorenz 96 Model) A comparison between the different filters. Details as in \autoref{fig:Linear_mean_errors} except that top-left panel shows the true signal \eqref{eq:signal} as reference.}
\label{fig:L96_mean_errors}
\end{figure}
\begin{figure}[h!]
\centering
\includegraphics[width =0.7\textwidth, height = 0.43\textheight]{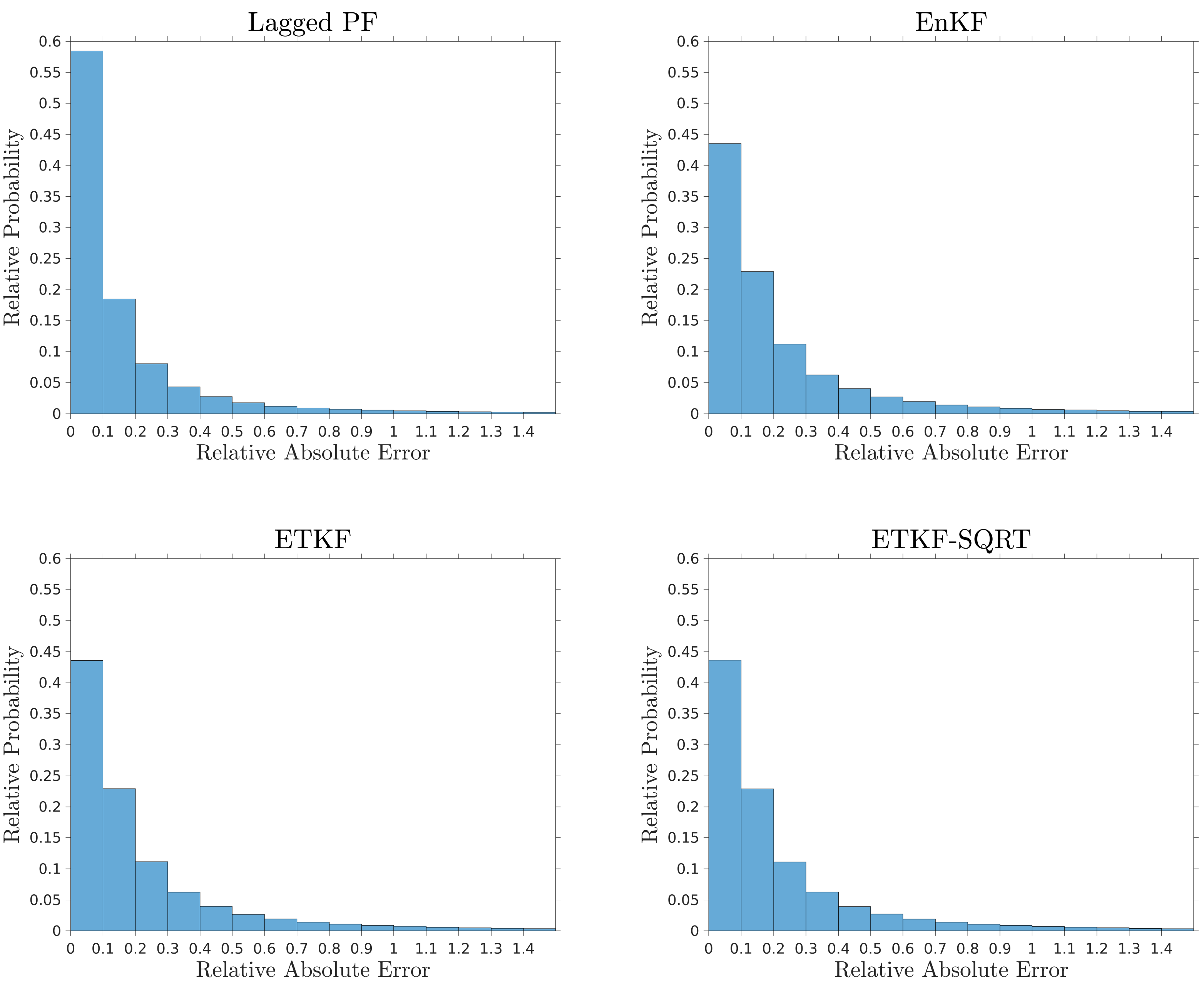}

\caption{(Lorenz 96 Model) A histogram of the relative absolute errors for each filter. Details as in \autoref{fig:Linear_histo}. 
}
\label{fig:L96_histo}
\end{figure}
\begin{figure}[h!]
\centering
\includegraphics[width =0.98\textwidth]{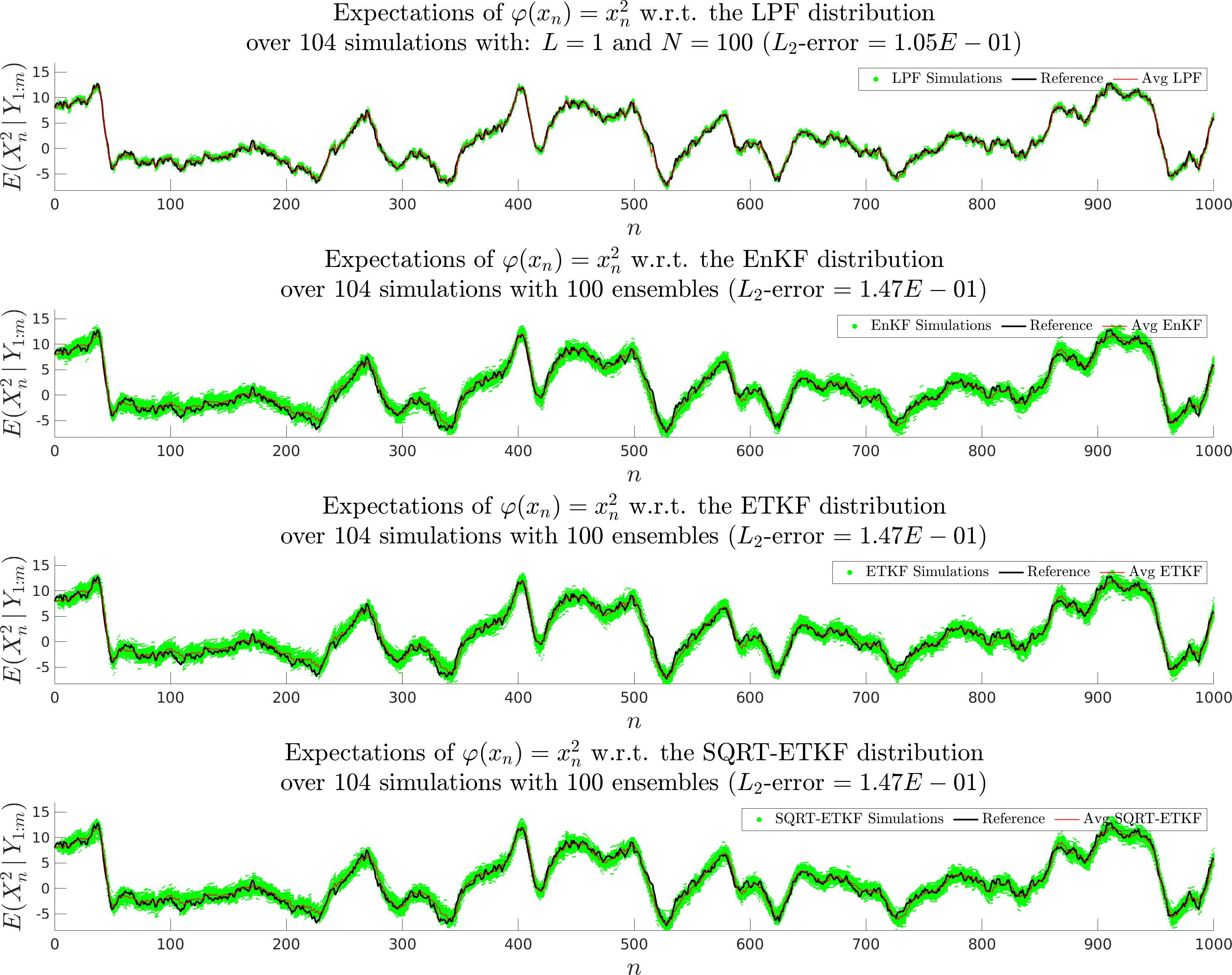}
\caption{(Lorenz 96 Model) Approximation clouds of $\pi_n(\varphi)$ against $n$ with $\varphi(x_{1:n})=x_n^{2}$ for each method where the observations arrive every 3 time steps. Details as in \autoref{fig:Linear_1coord} except that in the 
black curve the true signal is plotted as the reference.
}
\label{fig:L96_1coord}
\end{figure}
The function $q_n(X_{n-1})$ is given by the 4th-order Runge Kutta approximation of the Lorenz 96 system 
\begin{align}
\frac{dX_{n-1}^{i}}{dt} = X_{n-1}^{i-1}(X_{n-1}^{i+1} - X_{n-1}^{i-2}) - X_{n-1}^{i} + 8, \label{eq:L96}
\end{align}
when run one step in time with step size $\Delta_t$. 
Here $X_{n-1}^{i}$ is the $i$th-component of $X_{n-1}$ with periodic conditions, $X_{n-1}^{0}=X_{n-1}^{d}$, $X_{n-1}^{-1}=X_{n-1}^{d-1}$ and $X_{n-1}^{d+1}=X_{n-1}^{1}$.

For the numerical results, we set $\mu_{n-L}(x_{n-L+1})=\mathcal{N}(x_{n-L+1}; m_{n-L+1},P_{n-L+1})$ for $n\geq L+1$, where $m_{n-L+1}$ and $P_{n-L+1}$ are the mean and covariance of the 
ETKF-SQRT predictor distribution at the time step $n-L+1$. We set $N=100$, $N^* = 0.6 N$, $T=10^3$, $\hat{k}=3$ (i.e. $M=333$), $d=d_y=200$, $R_1^{1/2}=0.5 I_d$, $R_2^{1/2}=0.2 I_{d_y}$, $\Delta_t =0.01$, $L=1$, $C=I_d$, $X_0^{20}=8.10$, 
$X_0^{i}=8$ for $1\leq i<20$, $20<i\leq d$. 

Our results are presented similarly to before. The left and right panels of \autoref{fig:L96_mean_errors} show approximations of  $\pi_n(\varphi)$ with $\varphi(x_{1:n}) = x_n$ and the associated relative absolute errors for LPF, EnKF, ETKF and ETKF-SQRT. 
The reference in this example is the true signal from which the data was generated. 
In \autoref{fig:L96_histo}, we present as before a histogram of the relative errors and observe that for the same number of particles/ensembles about 59\% of the LPF errors are less than 0.1 and 43\% in the rest of the methods.
Finally \autoref{fig:L96_1coord} shows the approximations of $\pi_n(\varphi)$ against time with $1\leq n \leq T$ and $\varphi(x_{1:n}) = x_n^{2}$. We can see from this figure that the 104 runs of the LPF method have low variance (which corresponds to a thinner green cloud) and on average is very close to the 2nd coordinate of the true signal.

\subsubsection{Conservative Shallow-Water Model}
We apply the algorithm on the conservative shallow-water model \cite{SWE1}. Function $q_n(X_{n-1})$ is the finite volume (FV) solution of the shallow-water equations (SWE) given below (one step in time with step size $\Delta_t$). Let $(x,y)\in \Omega = [0,a]^2$, for some $a>0$, and let $h(t,x,y)$, $v(t,x,y)$ and $u(t,x,y)$ represent the fluid column height, the fluid's horizontal velocity in the $x$-direction and the fluid's horizontal velocity in the $y$-direction, respectively, at position $(x,y)$ and time $t>0$. The conservative form of SWE is as follows
\begin{align*}
&\frac{\partial h}{\partial t} + \frac{\partial (hu)}{\partial x} + \frac{\partial (hv)}{\partial y} = 0,\\
&\frac{\partial (hu)}{\partial t} + \frac{\partial}{\partial x} (hu^2 +\tfrac{1}{2}gh^2) + \frac{\partial (huv)}{\partial y}  = 0,\\
&\frac{\partial (hv)}{\partial t} +  \frac{\partial (huv)}{\partial y}  + \frac{\partial}{\partial y} (hv^2 +\tfrac{1}{2}gh^2) = 0,
\end{align*} 
where $g$ is the gravitational acceleration. To write the equations in a compact form, we introduce three vectors $U=[h,hu,hv]^{\top}$, $A(U) = [hu,hu^2 +\frac{1}{2}gh^2,huv]^{\top}$ and $B(U) = [hv,huv,hv^2 +\frac{1}{2}gh^2]^{\top}$. Under this notation, we have
\begin{align}
\frac{\partial U}{\partial t} + \frac{\partial A(U)}{\partial x} + \frac{\partial B(U)}{\partial y} = 0.\label{eq:SWE}
\end{align}
\begin{figure}[h!]
\centering
\begin{subfigure}[c]{0.5\textwidth}
\includegraphics[width = 0.8\textwidth, trim={2.3cm 14.5cm 7.3cm 5.3cm},clip]{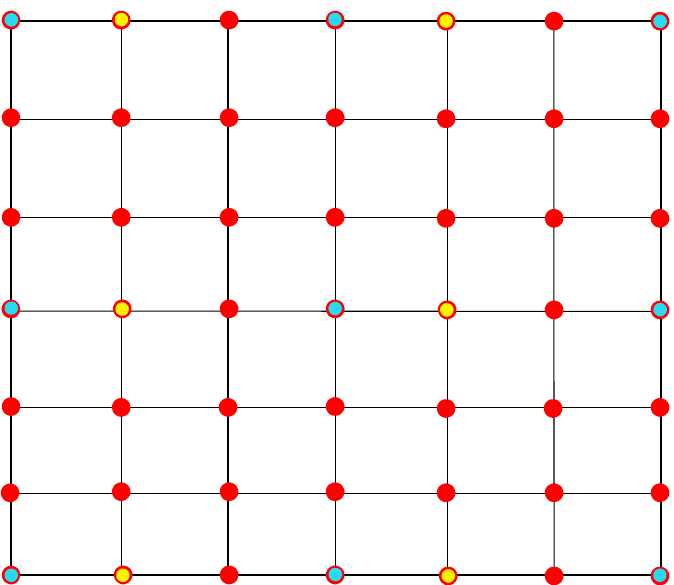}
\end{subfigure}
\begin{subfigure}[c]{0.45\textwidth}
\includegraphics[width =\textwidth, trim={0cm 0cm 0cm 2cm},clip]{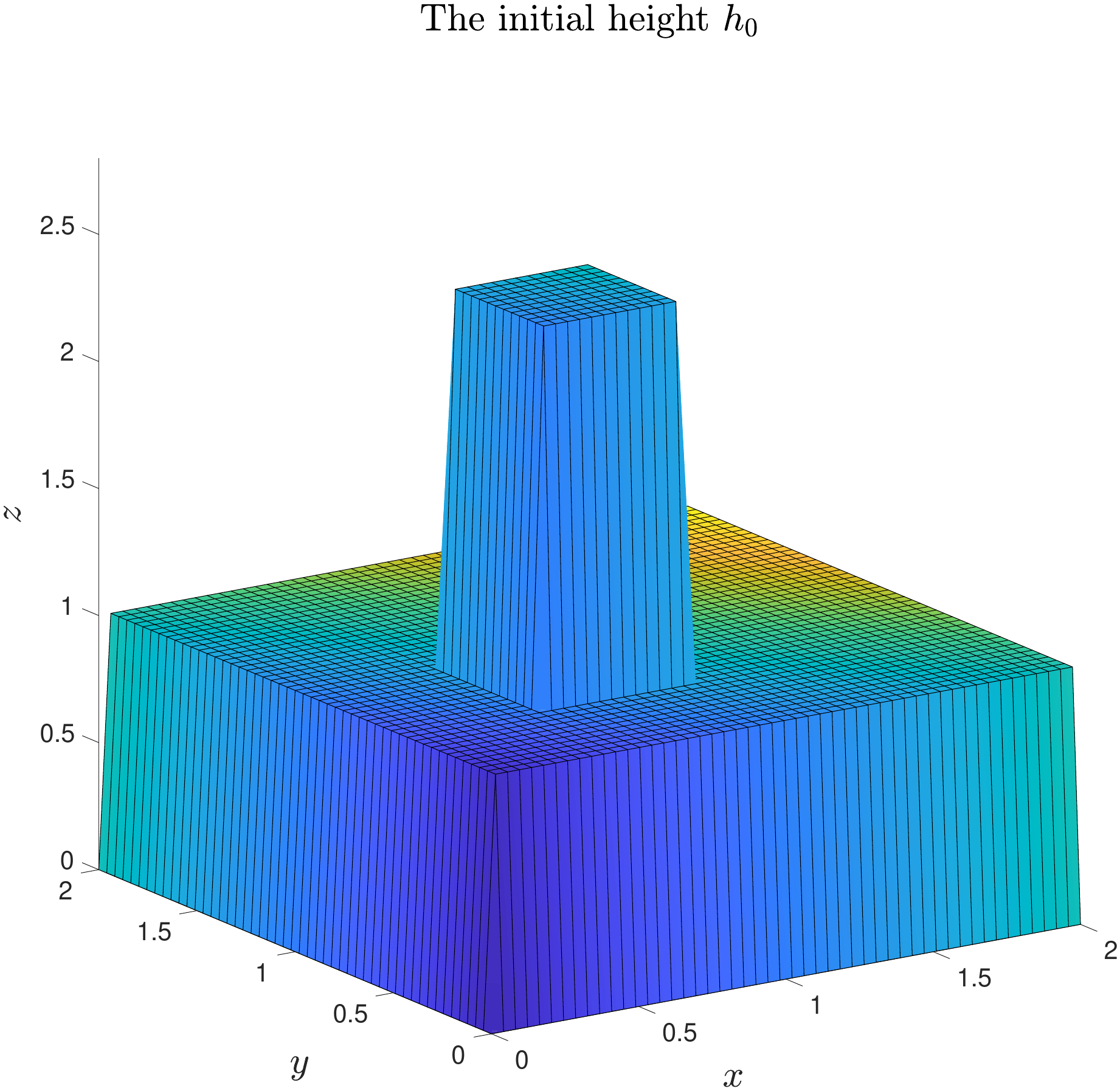}
\end{subfigure}
\caption{The figure on the left represents the physical grid where the red, cyan and yellow circles are the grid points at which $h$, $u$ and $v$ are observed, respectively. The figure on the right shows the initial height of the water.}
\label{fig:FV_grid}
\end{figure}

\begin{figure}[!h]
\centering
\includegraphics[width = 0.75 \textwidth]{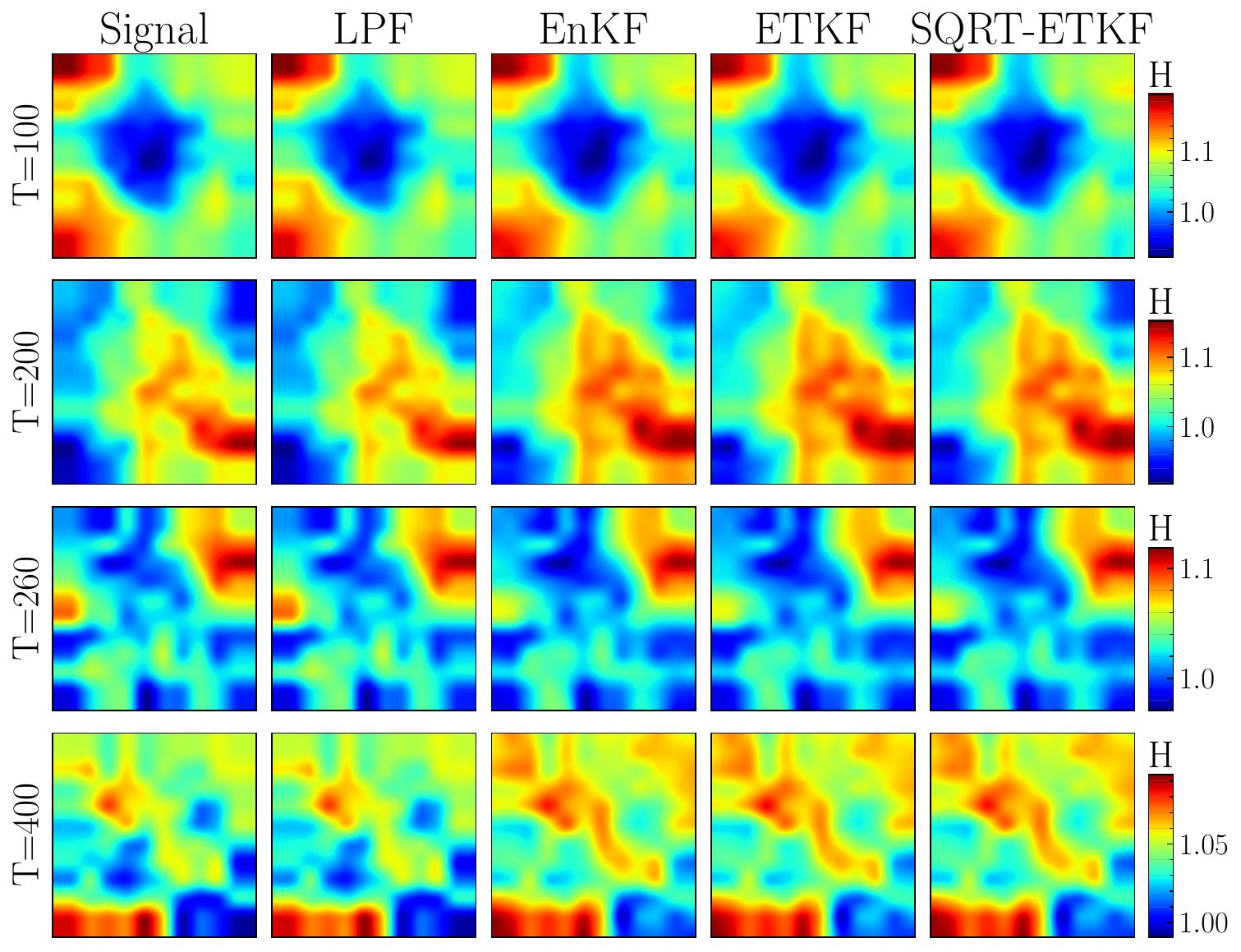}

\caption{(SWE Model) Snapshots of the height $h$ at four different times: 100 (first row), 200 (second row), 260 (third row) \& 400 (fourth row). From left to right: the signal, the mean of 50 simulations of the LPF with 100 particles and the mean of 50 simulations of the ENKF, ETKF \& ETKF-SQRT with 1000 ensembles, respectively.}
\label{fig:SWE_height}
\end{figure}

\begin{figure}[!h]
\centering
\includegraphics[width = 0.75\textwidth]{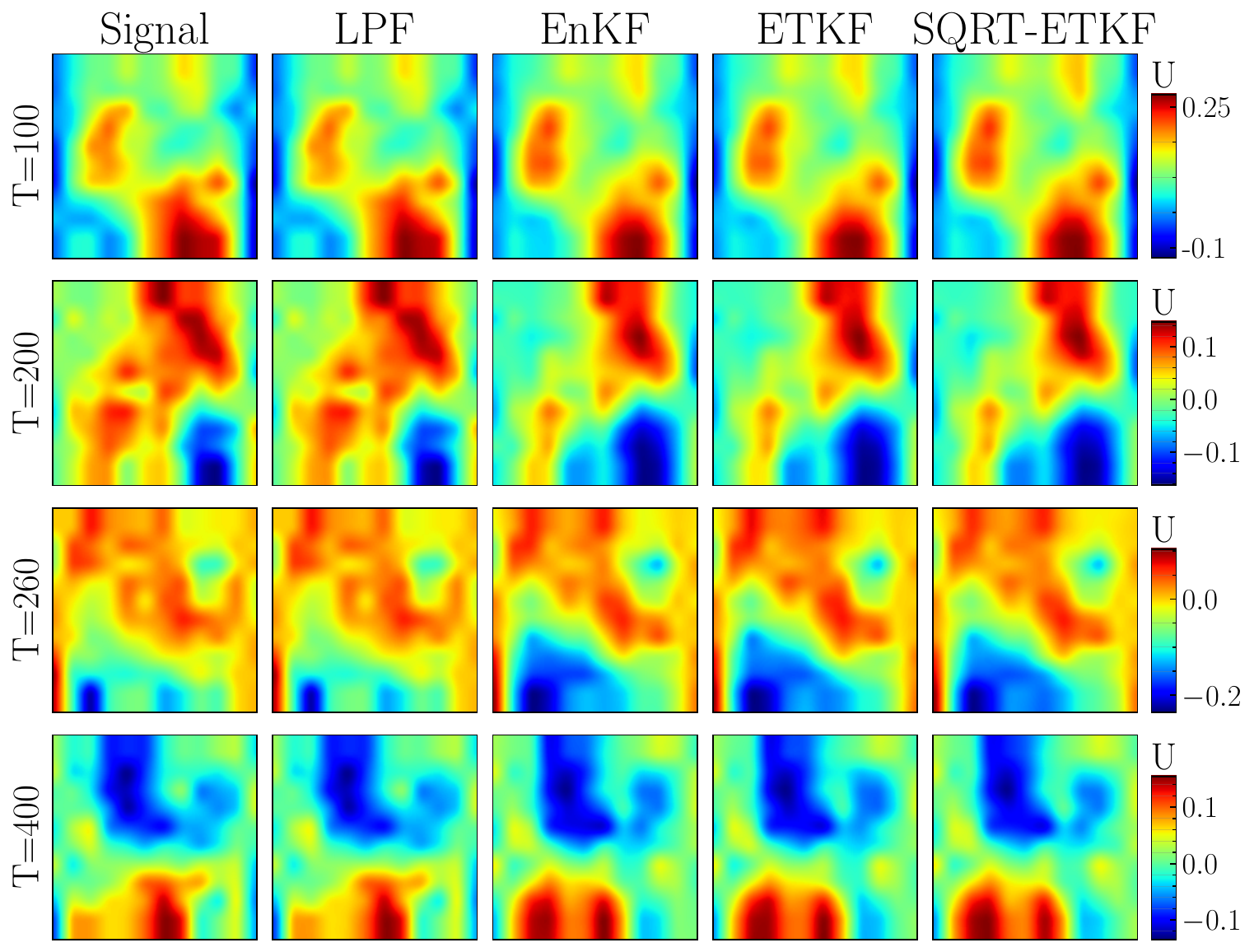}

\caption{(SWE Model) Snapshots of the horizontal $x$-axis velocity $u$ at four different times: 100 (first row), 200 (second row), 260 (third row) \& 400 (fourth row). From left to right: the signal, the mean of 50 simulations of the LPF with 100 particles and the mean of 50 simulations of the ENKF, ETKF \& ETKF-SQRT with 1000 ensembles, respectively.}
\label{fig:SWE_U}
\end{figure}

\begin{figure}[h!]
\centering
\includegraphics[width =0.75\textwidth]{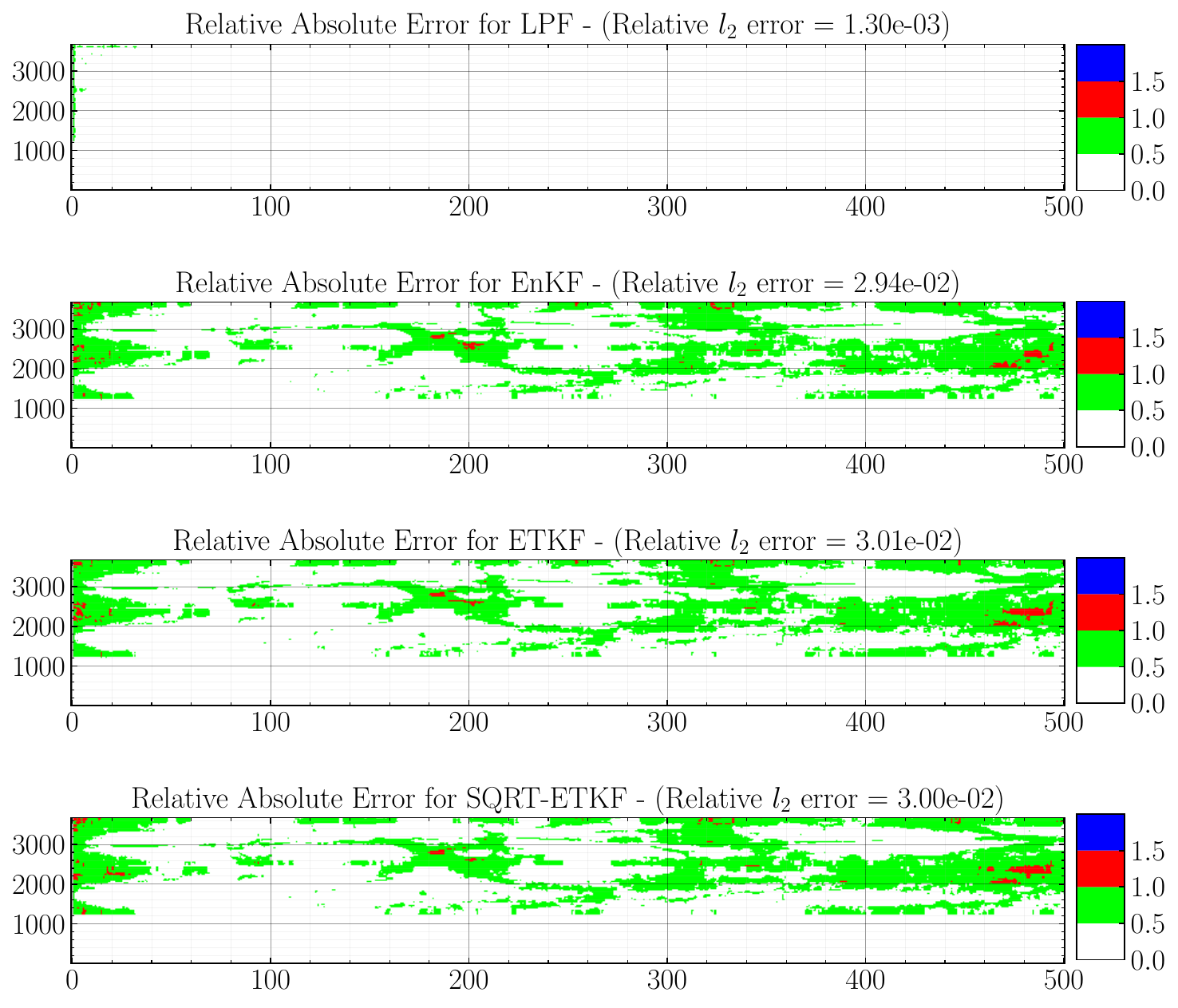}
\caption{(SWE Model) The relative absolute errors. The horizontal axis represents the time parameter $n$, while the vertical axis represents the signal's coordinates. We also include the relative $l_2$-norm in the titles.}
\label{fig:SWE_mean_abs_errors}
\end{figure}
\begin{figure}[h!]
\centering
\includegraphics[width = 0.7\textwidth]{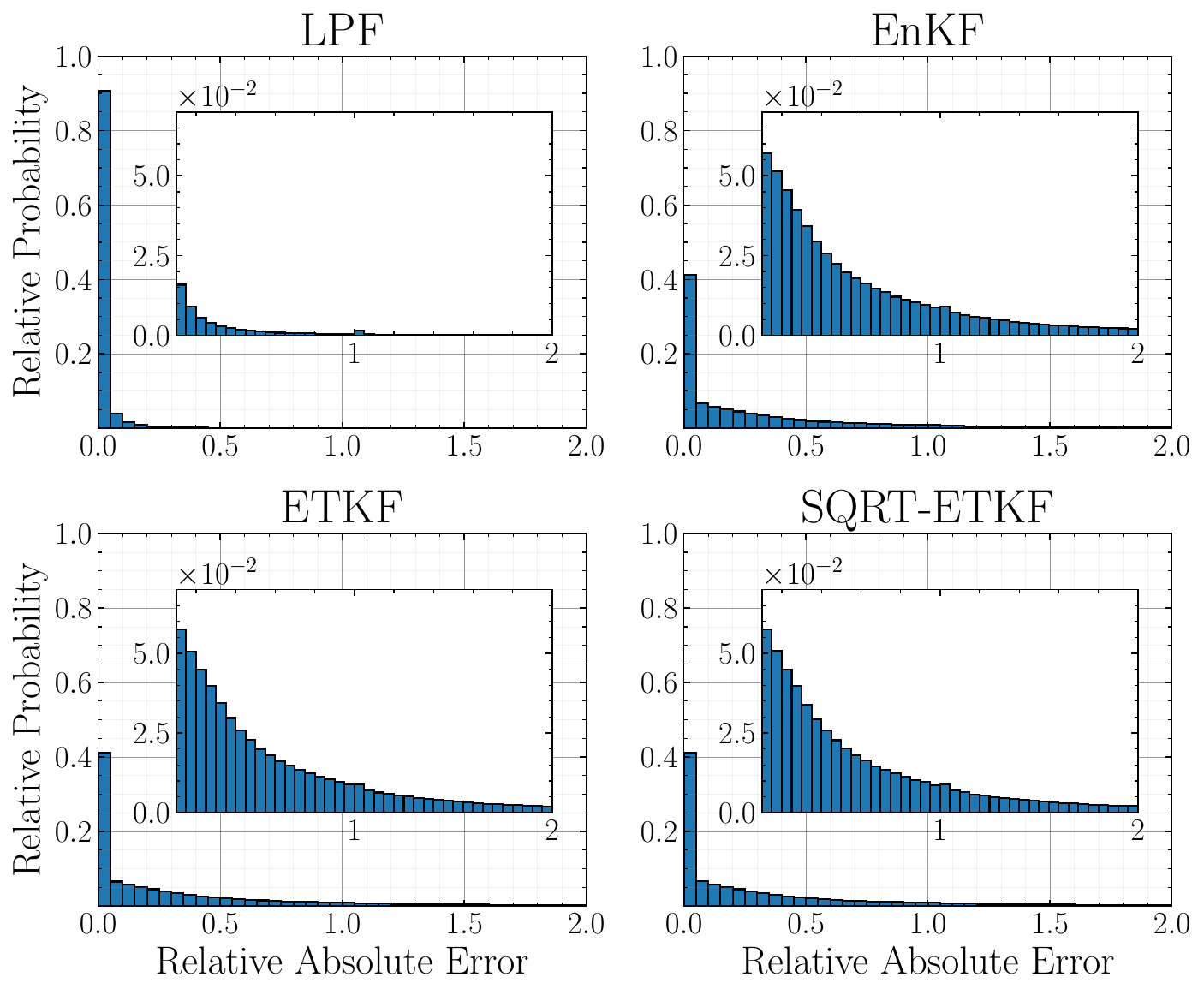}
\caption{(SWE Model) A histogram of the relative absolute errors for each filter. The relative probability here is defined as the number of elements in the bin divided by the total number of elements, which is $d \times (T+1)$. For readability, we zoom in the region $[0,2]\times [0,0.1]$ and include it in each histogram.}
\label{fig:SWE_hist}
\end{figure}
\begin{figure}[!t]
\centering
\includegraphics[width = 0.88\textwidth]{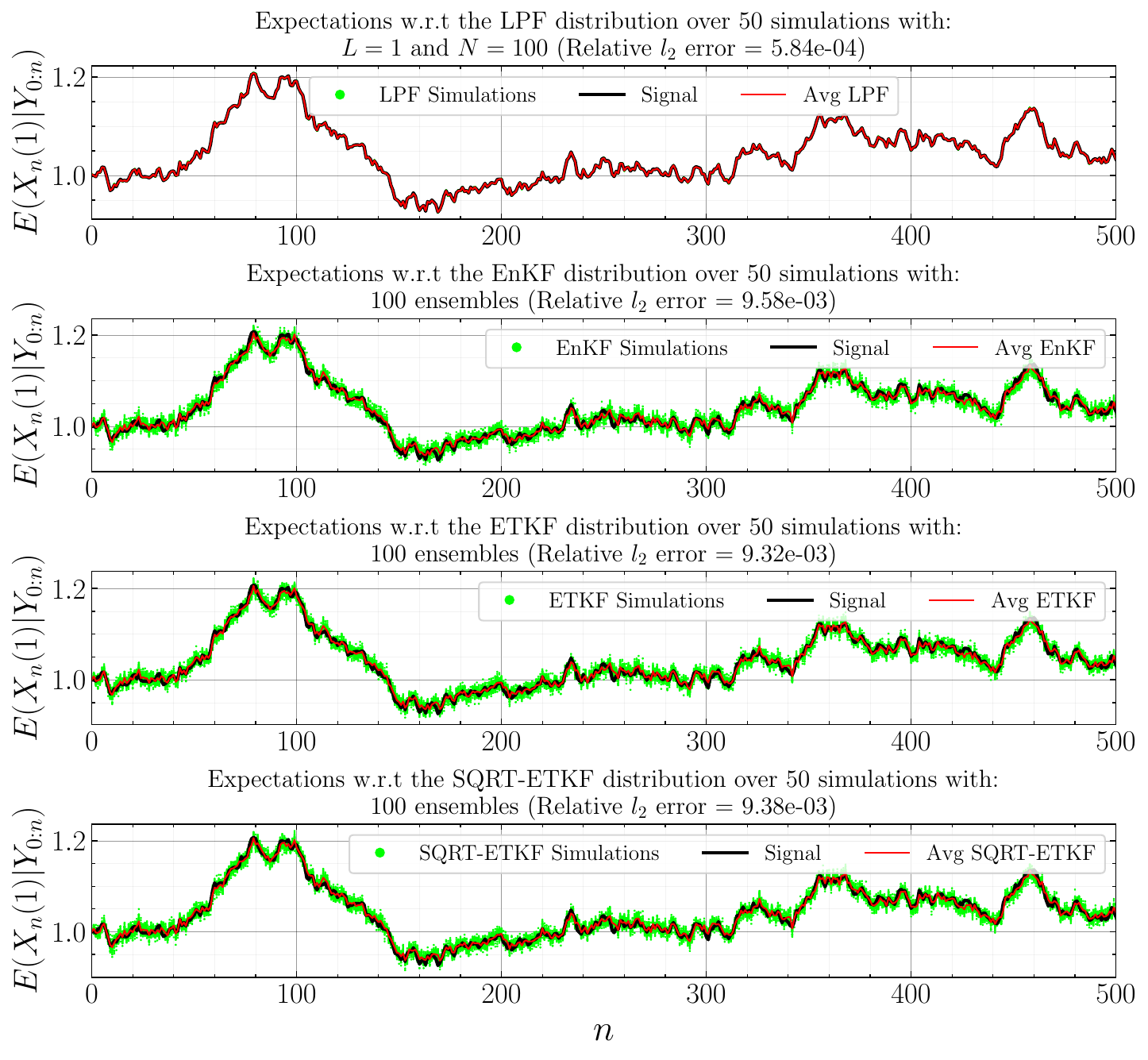}
\caption{(SWE Model) A comparison of the expectation of $\varphi(x_{1:n})=x_n^{1}$ as found by LPF, EnKF, ETKF and ETKF-SQRT. The green circles represent the different simulations of each filter. The black curve is the reference (in this example it is the signal) and the red curve is the average of the simulations.}
\label{fig:SWE_coord1}
\end{figure}
\begin{figure}[!t]
\centering
\includegraphics[width = 0.9 \textwidth]{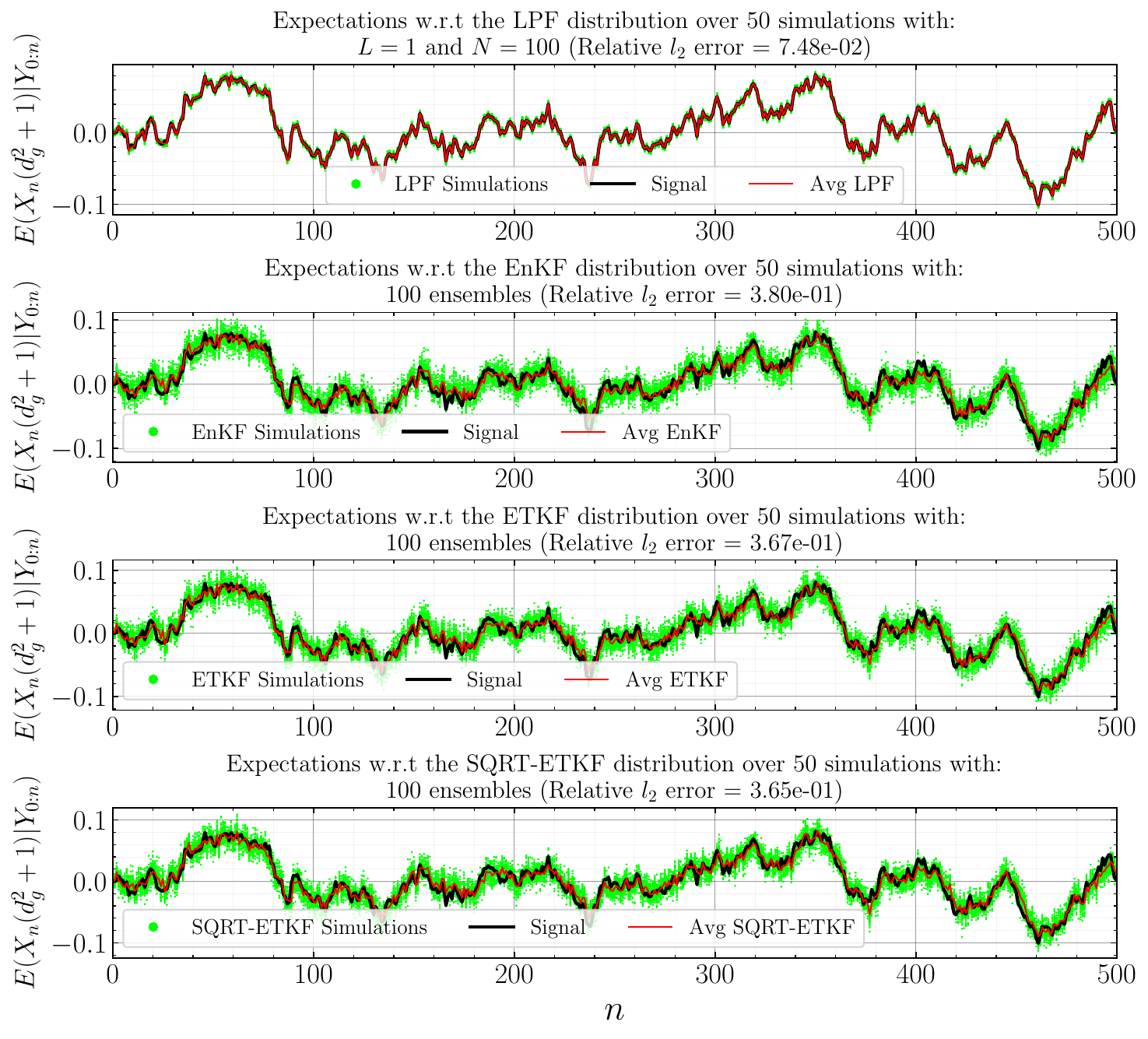}
\caption{(SWE Model) A comparison of the expectation of $\varphi(x_{1:n})=x_n^{d_g^2+1}$ as found by LPF, EnKF, ETKF and ETKF-SQRT. The green circles represent the different simulations of each filter. The black curve is the signal and the red curve is the average of the simulations.}
\label{fig:SWE_coord_dg2p1}
\end{figure}
Given the number of grid points in the $x$ and $y$ directions, $d_g$, we set $\Delta_x = \Delta_y = a/d_g$. We will refer to the grid $\big\{(x_i,y_j) \in \Omega:\, x_i = (i-1)\Delta_x, \, y_j = (j-1) \Delta_y, \, \text{for } i,j \in \{1,\ldots,d_g\}\big\}$ as the physical grid. Consider the uniform grid with FV cells $I_{i,j}=[x_{i-1/2},x_{i+1/2}] \times [y_{j-1/2},y_{j+1/2}]$ centered at $(x_i,y_j)=(\frac{x_{i-1/2}+x_{i+1/2}}{2},\frac{y_{j-1/2}+y_{j+1/2}}{2})$, for all $i,j \in \{0,\ldots, d_g+1\}$, with grid size of $(d_g+2)^2$. Then,
\begin{align*}
U_{i,j}^{n+1} = U_{i,j}^n - \frac{\Delta_t}{\Delta_x} (F_{i+\frac{1}{2},j}^n - F_{i -\frac{1}{2},j}^n) - \frac{\Delta_t}{\Delta_y} (G_{i,j+\frac{1}{2}}^n - G_{i ,j-\frac{1}{2}}^n),
\end{align*}
where $F$ and $G$ are the numerical Lax-Friedrichs fluxes given by
\begin{align*}
F_{i+\frac{1}{2},j}^n &= \tfrac{1}{2} [A(U_{i,j}^n)+A(U_{i+1,j}^n)]  - \tfrac{1}{2} \lambda_{i+\frac{1}{2},j,\max}^x [U_{i+1,j}^n - U_{i,j}^n]\\
F_{i-\frac{1}{2},j}^n &= \tfrac{1}{2} [A(U_{i,j}^n)+A(U_{i-1,j}^n)]  - \tfrac{1}{2} \lambda_{i-\frac{1}{2},j,\max}^x [U_{i,j}^n - U_{i-1,j}^n]\\
G_{i,j+\frac{1}{2}}^n &= \tfrac{1}{2} [B(U_{i,j}^n)+B(U_{i,j+1}^n)]  - \tfrac{1}{2} \lambda_{i,j+\frac{1}{2},\max}^y [U_{i,j+1}^n - U_{i,j}^n]\\
G_{i,j+\frac{1}{2}}^n &= \tfrac{1}{2} [B(U_{i,j}^n)+B(U_{i,j-1}^n)]  - \tfrac{1}{2} \lambda_{i,j-\frac{1}{2},\max}^y [U_{i,j}^n - U_{i,j-1}^n],\\
\end{align*}
where $\lambda_{i^*,j^*,\max}^x$ is the maximum eigenvalue of the Jacobian matrix $\partial A(U)/\partial U$ evaluated at $U_{i^*,j^*}^n$. The eigenvalues are $\{u_{i^*,j^*}\pm \sqrt{gh_{i^*,j^*}},u_{i^*,j^*}\}$. We set $\lambda_{i^*,j^*,\max}^x =|u_{i^*,j^*}|+\sqrt{gh_{i^*,j^*}}$. Similarly, $\lambda_{i^*,j^*,\max}^y$ is the maximum eigenvalue of the Jacobian matrix $\partial B(U)/\partial U$ evaluated at $U_{i^*,j^*}^n$ and we take it to be $|v_{i^*,j^*}|+\sqrt{gh_{i^*,j^*}}$. The initial value conditions are $h(0,x,y) = h_0(x,y)$ and $v(0,x,y) = u(0,x,y) = 0$ for $(x,y)\in \Omega$. We use reflective boundary conditions with $u(t,0,y)=u(t,a,y)=0$ and $v(t,x,0) = v(t,x,a) = 0$. The dimension of the state vector in \eqref{eq:signal} is $d=3d_g^2$. 

We write $(h_{i,j})_{1\leq i,j \leq d_g }$, $(u_{i,j})_{1\leq i,j \leq d_g }$ and $(v_{i,j})_{1\leq i,j \leq d_g }$ in a vector representation as $(h_i)_{1\leq i\leq d_g^2}$, $(u_i)_{1\leq i\leq d_g^2}$ and $(v_i)_{1\leq i\leq d_g^2}$. Then we write the state vector $X_n = (X_n^i)_{1\leq i\leq d}$ in \eqref{eq:signal} as
$$
X_n = [(h_{i})_{1\leq i \leq d_g^2 }, (u_{i})_{1\leq i \leq d_g^2 },(v_{i})_{1\leq i \leq d_g^2 } ]^\top \in \mathbb{R}^d.
$$
We take $a=2$, $d_g=35$ (so $d=3675$) and $h_0(x,y) = 2.5$ for $ 0.5 \leq x,y \leq 1$ and one elsewhere. The spatial frequency at which height is observed is 1, i.e., it is observed at all grid points, while the spatial frequency at which $u$ and $v$ are observed is $3$. However, the first component that is observed in $u$ is $X_n^{d_g^2+1}$ and the first component that is observed in $v$ is $X_n^{2d_g^2+2}$ (see \autoref{fig:FV_grid}). Therefore we take the matrix $C$ as
\begin{align*}
C = \left[\begin{array}{cccc|cccc|cccc}
\bm{e}_1 & \bm{e}_2 & \cdots & \bm{e}_{d_g^2}   & \bm{e}_{d_g^2+1} & \bm{e}_{d_g^2+4} & \bm{e}_{d_g^2+7} &
\cdots & \bm{e}_{2d_g^2+2} & \bm{e}_{2d_g^2+5} & \bm{e}_{2d_g^2+8} & \cdots
\end{array} \right]^\top,
\end{align*}
where $\bm{e}_\ell$ is a vector in $\mathbb{R}^{d}$ with one at position $\ell$ and zeros everywhere else. We set $N=100$, $N^*=0.5~N$, $T=500$, $\hat{k}=1$ (i.e. $M=500$), $R_1^{1/2}=0.01 I_d$, $R_2^{1/2}=0.01 I_{d_y}$ and 
$$
\Delta_t =\frac{\text{CFL} \times  \Delta_x}{\max\left\{\max_{0 \leq i,j \leq d_g+1}\left(\lambda_{i\pm \frac{1}{2},j,\max}^x \right),\max_{0 \leq i,j \leq d_g+1}\left(\lambda_{i,j\pm \frac{1}{2},\max}^y\right )\right\}},
$$ 
where CFL is the Courant-Friedrich-Levy number which we set to $0.5$. The approximation $\mu_{n-L}$ used here is a Gaussian with mean and covariance obtained from the prediction step of the EnKF with 1000 ensemble size.

\autoref{fig:SWE_height} and \autoref{fig:SWE_U} display, respectively, snapshots of the perturbed and filtered water's height and water's horizontal $x$-axis velocity at times 100, 200, 260 \& 400 
using the different methods. We ran 50 independent simulations of each of the EnKF, ETKF and ETKF-SQRT methods with 1000 ensembles because they provided unreliable estimates for lower ensemble sizes. We ran 50 simulations of the LPF with only 100 particles. \autoref{fig:SWE_mean_abs_errors} shows the relative absolute errors for each filter. In \autoref{fig:SWE_hist} a histogram of the relative absolute errors is provided. 
From the histograms we note that around 91\% of the relative absolute errors of the LPF are less than 0.01 whereas in the other methods they are around 41\%. Finally, \autoref{fig:SWE_coord1} and \autoref{fig:SWE_coord_dg2p1} show 
respectively cloud plots of the approximations of 
$\mathbb{E}\left(x_n^1|y_{1:n}\right)$ and $\mathbb{E}\left(x_n^{d_g^2+1}|y_{1:n}\right)$ for $1\leq n \leq T$. As in the linear Gaussian example the variability is much less for the LPF.
All plots show that the LPF with only 100 particles outperforms the other methods in terms of accuracy and this is illustrated by the one order of magnitude improvement in terms of relative $L_2$ errors.

\section{Discussion}

In the literature fixed-lag methods have been mainly in the context of improving smoothing approximations. 
The original motivation was to alleviate path degeneracy and justify biased approximations that resample only paths within a fixed lag to the current time, \cite{kita,olsson}. 
Later these ideas have been complemented with MCMC for the purpose of parameter estimation \cite{polson_practical}. In addition, SMC sampling techniques 
for a fixed lag have been used without bias to enhance particle filtering \cite{doucet06fixedlag}. All these contributions are very close to the spirit of this paper, 
but to the best of our knowledge this paper is the first to use fixed lag sampling ideas for high dimensional SSMs. 

We have proposed a fixed-lag particle filter that under certain conditions on the model has provably stable errors as the dimension of the model grows. 
Our proposed algorithm also incorporates adaptive tempering steps, which have been crucial 
for other successful SMC approaches for high-dimensional problems \cite{llopis,cotter}. Tempering was also crucial in the theory of \cite{beskos,beskos1}, which was originally developed for independent states. 
The conditions used in this paper are general enough to go beyond this and allow for dependency between some coordinates of the states or observations and at the same time be able to 
apply the results in \cite{beskos,beskos1}. As mentioned earlier, our algorithm  inherently contains a bias as other similar fixed-lag approaches, 
but we have shown this bias to be bounded and controlled by the lag parameter $L$ and the number of particles $N$. 
Regarding performance, that of the proposed method was found to be superior to commonly used methods such as the EnKF, ETKF in our numerical examples. 
In particular the results in the challenging shallow water equation showed impressive improvement and accuracy. Future work can include more case studies 
in even more challenging applications and comparisons with current state-of-art methods.

\subsubsection*{Acknowldegements}

AJ \& HR were supported by KAUST baseline funding. The work of DC has been partially supported by European Research Council (ERC) Synergy grant STUOD-DLV-8564. NK was supported by a J.P. Morgan A.I. Research Award.
\appendix

\section{Technical Results}

For completeness, we describe briefly a result  that expresses in a more general formulation  findings from the works in  \cite{beskos, beskos2}. 
Consider a sequence of distributions on the space  $\mathsf{E}^{\mathsf{L}d}$, with $\mathsf{E}\subset \mathbb{R} $ assumed compact and a fixed $\mathsf{L}\in \mathbb{N}$. We define the bridging sequence of densities:
\begin{align*}
\mathsf{\Pi}_k := \widetilde{\mathsf{\Pi}}^{1-\phi_k}\times  \mathsf{\Pi}^{\phi_k}, \qquad \phi_1=0, \quad \phi_{k+1}-\phi_k = \tfrac{1}{d}, \quad k=1,2,\ldots, d+1,
\end{align*}
The initial and target densities, $\widetilde{\mathsf{\Pi}}:\mathsf{E}^{\mathsf{L}d}\to\mathbb{R}$ and $\mathsf{\Pi}:\mathsf{E}^{\mathsf{L}d}\to \mathbb{R}$ respectively,  admit the factorisation, for fixed $m\in \mathbb{N}$, 
$\widetilde{\mathsf{\Pi}}^{(m)}, 
\mathsf{\Pi}^{(m)}\in\mathcal{P}(\mathsf{E}^{\mathsf{L}m})$,  and
$\widetilde{\mathsf{\Pi}}^{1}, \mathsf{\Pi}^{1}\in \mathcal{P}(\mathsf{E}^{\mathsf{L}})$:
\begin{align*}
\widetilde{\mathsf{\Pi}}(x_{1:\mathsf{L}}) &= \widetilde{\mathsf{\Pi}}^{(m)}\big(x_{1:\mathsf{L}}^{1:m}\big)\times \prod_{j=m+1}^{d} \widetilde{\mathsf{\Pi}}^{1}(x_{1:\mathsf{L}}^{j}),\\ 
\Pi(x_{1:\mathsf{L}}) &= \mathsf{\Pi}^{(m)}\big(x_{1:\mathsf{L}}^{1:m}\big)\times \prod_{j=m+1}^{d} \mathsf{\Pi}^{1}(x_{1:\mathsf{L}}^{j}).
\end{align*}
Note that we have indexed the $\mathsf{L}d$ co-ordinates in an array format, using two indicators, one subscript for the `block', $l=1,\ldots,\mathsf{L}$, and one superscript for placement in a block, $j=1,\ldots, d$.
Consider related Markov kernels $\mathsf{K}_k:\mathsf{E}^{\mathsf{L}d}\times \mathcal{B}(\mathsf{E}^{\mathsf{L}d})\to   [0,1]$, $k=1,\ldots, d+1$, such that $\mathsf{\Pi}_k\mathsf{K}_k=\mathsf{\Pi}_k$. Also, kernels admit a factorisation, for $\mathsf{K}_{k}^{(m)}:\mathsf{E}^{\mathsf{L}m}\times \mathcal{B}(\mathsf{E}^{\mathsf{L}m})\to   [0,1]$ and $\mathsf{K}^{1}_{k}:\mathsf{E}^{\mathsf{L}}\times \mathcal{B}(\mathsf{E}^{\mathsf{L}})\to [0,1]$:
\begin{align*}
\mathsf{K}_k\big( x_{1:\mathsf{L}}, dz_{1:\mathsf{L}}\big) := \mathsf{K}_{k}^{(m)}
\big( x_{1:\mathsf{L}}^{1:m}, dz_{1:\mathsf{L}}^{1:m}\big) \times  \prod_{j=m+1}^d \mathsf{K}^{1}_{k}\big(x_{1:\mathsf{L}}^{j}, dz_{1:\mathsf{L}}^{j}  \big),
\end{align*}
so that  $\mathsf{\Pi}_{k}^{(m)} \mathsf{K}_{k}^{(m)} = \mathsf{\Pi}_{k}^{(m)}$, 
$\mathsf{\Pi}_{k}^{1} \mathsf{K}_{k}^{1} = \mathsf{\Pi}_{k}^{1}$, with 
$\mathsf{\Pi}_{k}^{(m)}$, $\mathsf{\Pi}_{k}^{1}$ defined in an obvious way.
One now applies the SMC sampler in \autoref{alg:smc_samp}, in the setting we have formulated herein, under the more general consideration that the initial positions from the $N$ particles are chosen arbitrarily from within the state space $\mathsf{E}^{\mathsf{L}d}$. The sequence of importance sampling  weights are given as:
\begin{align*}
W_{k+1}^{(i)} = W_{k}^{(i)} \times    \big(\tfrac{\mathsf{\Pi}}{\widetilde{\mathsf{\Pi}}}\big)\big(x_{1:\mathsf{L},k}^{(i)}\big)      ^{\phi_{k+1}-\phi_k}, \qquad k=2,\ldots, d, \quad i=1,2,\ldots, N.
\end{align*}
We denote here by $x_{l,k}^{(i)}$ the $l$th `block', $l=1,\ldots, \mathsf{L}$, of 
$i$th particle after the application of kernels $\mathsf{K}_2, \ldots, \mathsf{K}_k$, $k=2,\ldots, d+1$.
Simple calculations give that:
\begin{align*}
\log W_{d+1}^{(i)} &=  \tfrac{1}{d}\sum_{k=1}^{d} \log\Big(\big(\tfrac{\mathsf{\Pi}}{\widetilde{\mathsf{\Pi}}}\big)\big(x_{1:\mathsf{L},k}^{(i)}\big)\Big)\\
&=  \tfrac{1}{d}\sum_{k=1}^{d} \log\Big(\big(\tfrac{\mathsf{\Pi}^{(m)}}{\widetilde{\mathsf{\Pi}}^{(m)}}\big)\big(x_{1:\mathsf{L},k}^{(i),1:m}\big)\Big) +  \tfrac{1}{d}\sum_{k=1}^{d}\sum_{j=m+1}^{d} \log\Big(\big(\tfrac{\mathsf{\Pi}^{1}}{\widetilde{\mathsf{\Pi}}^{1}}\big)\big(x_{1:\mathsf{L},k}^{(i),j}\big)\Big).
\end{align*}
Notice the two above sums involve two \emph{independent} Markov sequences of random variables. 

The work in  \cite{beskos} can now provide the following results. 
We adopt the notation for a \emph{continuum} of involved distributions and Markov kernels, by using subscripts $s,t\in[0,1]$ below, at the same position that we have so far used $n\in\mathbb{N}$. Also,  
the $\sigma$-algebra $\mathcal{F}_0$ used below corresponds to the information about the initial position of the Markov chains. 

\begin{ass}
\label{ass:G}
Consider the following assumptions:
\begin{itemize}
\item[(A1)]
There exist $\theta\in(0,1)$,  $(\zeta^{(m)},\zeta)\in\mathscr{P}(\mathcal{B}(\mathsf{E}^{m \mathsf{L}}))\times
\mathscr{P}(\mathcal{B}(\mathsf{E}^{\mathsf{L}}))$, such that for each $s\in[0,1]$:
%
%
%
%
\begin{align*}
\mathsf{K}_{s}^{(m)}
\big(x_{1:\mathsf{L}}^{1:m},A\big) & \geq  \theta\,\zeta^{(m)}(A),\quad  \textrm{for all }(x_{1:\mathsf{L}}^{1:m},A)\in \mathsf{E}^{m\mathsf{L}}\times\mathcal{B}(\mathsf{E}^{m\mathsf{L}}), \\[0.2cm]
\mathsf{K}_{s}(x_{1:\mathsf{L}},A) & \geq  \theta\,\zeta(A),\quad \textrm{for all } (x_{1:\mathsf{L}},A)\in\mathsf{E}^{\mathsf{L}}\times\mathcal{B}(\mathsf{E}^{\mathsf{L}}).
\end{align*}
\item[(A2)]
There exists $C<\infty$ such that
for any $s,t\in[0,1]$ we have
\begin{align*}
\sup_{x_{1:\mathsf{L}}^{1:m}\in \mathsf{E}^{m\mathsf{L}}}\big\|\mathsf{K}_{s}^{(m)}(x_{1:\mathsf{L}}^{1:m},\cdot)-\mathsf{K}_{t}^{(m)}(x_{1:\mathsf{L}}^{1:m},\cdot)\big\|_{\mathrm{TV}} & \leq  C|s-t|,\\
\sup_{x_{1:\mathsf{L}}\in \mathsf{E}^{\mathsf{L}}}\big\|\mathsf{K}_{s}(x_{1:\mathsf{L}},\cdot)-\mathsf{K}_{t}(x_{1:\mathsf{L}},\cdot)\big\|_{\mathrm{TV}} & \leq  C|s-t|.
\end{align*}
\item[(A3)]
There exists a $C<\infty$ such that:
$$\max\left\{\sup_{x^{1:m}_{1:\mathsf{L}}\in \mathsf{E}^{m\mathsf{L}}}\left|\log\left(\frac {\mathsf{\Pi}^{(m)}}{\widetilde{\mathsf{\Pi}}^{(m)}}\right)(x_{1:\mathsf{L}}^{1:m})\right|,\sup_{x_{1:\mathsf{L}}\in 
\mathsf{E}^{\mathsf{L}}}
\left|\log\left(\frac {\mathsf{\Pi}^{1}_1}{\widetilde{\mathsf{\Pi}}_1}\right)(x_{1:\mathsf{L}})\right|\right\}\leq C.$$
\end{itemize}

\end{ass}

\begin{theorem}
\label{th:G}
Under \autoref{ass:G}:
\begin{itemize}
\item[(i)] We have
$\log W_{d+1}^{(i)}  = o(1) + \overline{W}_{d+1}^{(i)}$,
%
where $o(1)$ denotes here a term that converges weakly to~$0$ as $d\rightarrow\infty$, and  we have defined: 
\begin{align*}
\log \overline{W}_{d+1}^{(i)}  &=\frac{1}{d}\sum_{k=1}^{d}\sum_{j=m+1}^{d} \left\{\, \log\left(\left(\frac{\mathsf{\Pi}^{1}}{\widetilde{\mathsf{\Pi}}^1}\right)\left(x_{1:\mathsf{L},k}^{(i),j}\right)\right)  - 
\mathbb{E}\,\left[\,\log\left(\left(\frac{\mathsf{\Pi}^{1}}{\widetilde{\mathsf{\Pi}}^{1}}\right)\left(x_{1:\mathsf{L},k}^{(i),j}\right)\right)  \,\Big|\, \mathcal{F}_0 \right]\, \right\}. 
\end{align*}
The following weak limit holds, as $d\rightarrow \infty$:
\begin{align*}
\log \overline{W}_{d+1}^{(i)}  \rightarrow_D \mathcal{N}(0, \sigma^{2}),
\end{align*}
for variance 
$$\sigma^2 = \int_{0}^{1}\Big\{\mathsf{\Pi}^{1}_s\big((\widehat{\varphi}_s)^2\big) - \big(\mathsf{\Pi}^{1}_s(\widehat{\varphi}_s)\big)^2\Big\}ds,$$ where $\widehat{\varphi}_s$ is the solution to the Poisson equation:
\begin{align*}
\varphi_s-\mathsf{\Pi}^{1}_s(\varphi_s) = \widehat{\varphi}_s - \mathsf{K}_{s}^{1}(\widehat{\varphi}_s), 
\end{align*}
for $\varphi_s = \varphi_s(x_{1}^{j},\ldots, x_{\mathsf{L}}^{j}) =   s\cdot \log\Big(\big(\tfrac{\mathsf{\Pi}^{1}}{\widetilde{\mathsf{\Pi}}^1}\big)\big(x_{1:\mathsf{L}}^{j}\big)\Big) $.
\item[(ii)] We have the weak limit, as $d\rightarrow\infty$, for any fixed $j^{*}\ge m+1$:
\begin{align*}
 x_{1:\mathsf{L},d+1}^{(i),j^*}  \rightarrow_D  
\mathsf{\Pi}^{1},
\end{align*}
and the one:
\begin{align*}
x_{1:\mathsf{L},d+1}^{(i),1:m} \rightarrow_D  
\mathsf{\Pi}^{(m)}.
\end{align*}

\end{itemize}
The weak limits above are jointly independent over particle index $i=1,\ldots, N$,  and  over the set of co-ordinates $\{1:m\}, m+1, \ldots, M^{\ast}$, for any fixed $M^{\ast}>m$.
\end{theorem}
\begin{proof}
The result is contained in Appendix C of \cite{beskos}. The crux of the proof for part (i) is the use of a CLT for triangular martingale arrays (notice that $\{\log \overline{W}_{k}^{(i)}\}_{k=1}^{d+1}$ is constructed to be a martingale process, for any $d\in \mathbb{N}$) and cumbersome but otherwise straightforward control of remainder terms. 
The term in the weights that involve the first $m$ co-ordinates disappear in the limit $d\rightarrow 0$, as there is only a finite number of them. Thus, the limiting values for the `standardised' weights $\overline{W}_{d+1}^{(i)}$ involve only CLT
obtained via co-ordinates $m+1,\ldots$. Result (ii) is obtained also in Appendix C of \cite{beskos}. Under the dynamics of the inhomogeneous Markov chain $\{\mathsf{K}_k\}$, the particles carry out enough steps to reach the correct distribution after the execution of $d$ steps, as $d\rightarrow\infty$.
\end{proof}

\section{Proof of Proposition \ref{prop:1}}
\label{sec:proofP}

Proposition \ref{prop:1} is a corollary of Theorem \ref{th:G}. In particular, it suffices to apply Theorem \ref{th:G} for each time index $n\in\mathbb{N}$ of interest, under the choices:

\begin{gather*}
\mathsf{L}=L+1, \quad \mathsf{\Pi} = \widehat{\pi}_{n}(dx_{(n-L):n}), \quad 
\widetilde{\mathsf{\Pi}} = \widehat{\pi}_{n-1}(dx_{(n-L):(n-1)})\times f(x_{n-1},dx_n), \\[0.2cm]  \mathsf{K}_s^{(m)}= \overline{K}^{(m)}_{s,n}, \quad \mathsf{K}_s= 
\overline{K}_{s,n}.
\end{gather*}

\noindent 
In the above expressions $\widehat{\pi}_{n}(dx_{(n-L):n})$ denotes the marginal of $\widehat{\pi}_{n}$ on components $x_{(n-L):n}$; and similarly for $\widehat{\pi}_{n-1}(dx_{(n-L):(n-1)})$.
The particular upper bounds  in the $L_2$-norms quoted in Proposition \ref{prop:1} follow immediately from the weak limits obtained via Theorem \ref{th:G} via relatively simple calculations that make use of the exchangeability of the limiting laws over the particle index. For the exact calculation, see Appendix A.2 of \cite{beskos1}.

\section{Algorithm with Adaptive Resampling and Tempering}

\begin{center}
\captionsetup[algorithm]{style=algori}
 \captionof{algorithm}{Lagged Particle Filter with Adaptive Resampling}
\label{alg:LPF}
\begin{itemize}
\item \textbf{Initialization $n=1$: }

We are given the annealing parameter $\phi_{1,1}\ge 0$, $N^*$ (the resampling threshold), the number of MCMC steps $S$, the lag $L+1$, and $x_0$. We begin by sampling particles $\tilde{x}_1^{(i)}$, $i\in\{1,\ldots,N\}$, from the proposal $\eta_1(x_1) = f(x_0,x_1)$ via \eqref{eq:signal}. We will sequentially sample from 
\begin{align*}
\tilde{\pi}_1^k(x_{1}) \propto \tilde{\gamma}_1^k(x_{1}) =f(x_0,x_1)\,  g(x_1,y_1)^{\phi_{k,1}}
\end{align*}
starting at $\tilde{\pi}_1^1$ ending at $\tilde{\pi}_1^{K_1}=\pi_1$ for some $K_1\in\mathbb{N}$. For $i\in\{1,\ldots,N\}$, compute the weights
\begin{align*}
W_{1,1}^{(i)} = \frac{\ \tilde{\gamma}_1^1(\tilde{x}_1^{(i)}) }{\eta_1( \tilde{x}_1^{(i)} )}= g(\tilde{x}_1^{(i)},y_1 )^{\phi_{1,1}}
\end{align*}
and normalize $\bar{W}_{1,1}^{(i)} = W_{1,1}^{(i)}/\sum_{j=1}^N W_{1,1}^{(j)}$.
Calculate the effective sample size (ESS): $ESS =\left( \sum_{i=1}^N (\bar{W}_{1,1}^{(i)})^2\right)^{-1}$.
If $ESS = ESS(\phi_{1,1}) \leq N^*$, resample $\{\tilde{x}_1^{(i)}\}_{i=1}^N$ and set $\bar{W}_{1,1}^{(i)}=\frac{1}{N}$, $i\in\{1,\ldots,N\}$. 

\noindent Set $k=1$ and denote the particles by $\{x_{k,1}^{(i)}\}$ and the weights by 
$\{\bar{W}_{k,1}^{(i)}\}$. Set $\textit{flag}=0$. \\ Perform the following SMC sampler:

While $\textit{flag}=0$, do:
\begin{enumerate}
\item Find $\delta \in [0,1-\phi_{k,1}]$ so that $ESS(\delta)=N^*$, where $ESS(\delta)$ is as above except now we have the unknown $\delta$ instead of $\phi_{1,1}$ in the weights. If $\delta=1-\phi_{k,1}$, set $\phi_{k+1,1}=1$, $K_{1}=k+1$ and $\textit{flag}=1$, else set $\phi_{k+1,1}=\phi_{k,1} + \delta$. 

\item Compute the weights
\begin{align*}
W_{k+1,1}^{(i)} &= W_{k,1}^{(i)}\cdot  \frac{ \tilde{\gamma}_1^{k+1}( x_{k,1}^{(i)} ) }{ \tilde{\gamma}_1^k( x_{k,1}^{(i)} )  }  = W_{k,1}^{(i)}\cdot 	g(x_{k,1}^{(i)},y_1)^{\delta},
\end{align*}
normalize $\bar{W}_{k+1,1}^{(i)} = W_{k+1,1}^{(i)}/\sum_{i=1}^N W_{k+1,1}^{(i)}$.
If $ESS \leq N^*$, resample $\{x_{k,1}^{(i)}\}$, set $\bar{W}_{k+1,1}^{(i)}=\frac{1}{N}$. Denote the samples  $\{\hat{x}_{k,1}^{(i)}\}$. 

\item (MCMC) For $m\in\{1,\ldots,S\}$, do: 
\begin{itemize}
\item For $i\in\{1,\ldots,N\}$, sample $x_{k+1,1}^{(i)}| \hat{x}_{k,1}^{(i)}$ from a Markov kernel that preserves
\[
\tilde{\pi}_1^{k+1}(x_1) \propto  f(x_0,x_1) g(x_1,y_1)^{\phi_{k+1,1}}.
\] E.g., one can use a random walk with proposal
\[x^{'(i)}_{1}= \hat{x}_{k,1}^{(i)} + Z^{(i)},\quad Z^{(i)} \sim \mathcal{N}(0,\Sigma_m),\]
for some covariance matrix $\Sigma_m$, so that the acceptance probability is
%
\begin{align*}
a_{m}^{(i)} = 1 \wedge \frac{ \tilde{\pi}_1^{k+1}(x^{'(i)}_1) } {  \tilde{\pi}_1^{k+1}(\hat{x}_{k,1}^{(i)}) } = 1\wedge \frac{   g(x^{'(i)}_{1},y_1)^{\phi_{k+1,1}} \, f(x_0,x^{'(i)}_1)  }{  g(\hat{x}_{k,1}^{(i)},y_1)^{\phi_{k+1,1}} \, f(x_0,\hat{x}_{k,1}^{(i)}) }.
\end{align*}
With probability $\alpha_m^{(i)}$, set $\hat{x}_{k,1}^{(i)}=x^{'(i)}_1$, $i\in\{1,\ldots,N\}$.
\end{itemize}
\item Set $k=k+1$ and $x_{k,1}^{(i)} = \hat{x}_{k-1,1}^{(i)}$, $i\in\{1,\ldots,N\}$.
\end{enumerate}

Calculate $ESS$ for the last  weights $\{\bar{W}_{K_1,1}^{(i)}\}$. If $ESS \leq N^*$, resample  $\{x_{K_1,1}^{(i)}\}$, set $\bar{W}_{K_1,1}^{(i)}=\frac{1}{N}$. Denote the final particles by $\{x_1^{(i)}\}$ and the final weights $\{\bar{W}_{1}^{(i)}\}$.

\item \textbf{Iterations $n = 2,3,\ldots,L$}: 

We have $N$ paths $\{x_{1:n-1}^{(i)}\}$ with weights $\{\bar{W}_{n-1}^{(i)}\}$. We sample $\tilde{x}_n^{(i)}$ from the proposal distribution $\eta_n(x_{1:n}) = \pi_{n-1}(x_{1:n-1})f(x_{n-1},x_n)$, via \eqref{eq:signal}. We will sequentially sample from 
\begin{align*}
\tilde{\pi}_n^k(x_{1:n}) \propto \tilde{\gamma}_n^k(x_{1:n}) :=
\Big(\prod_{j=1}^{n-1}g(x_j,y_j)\,f(x_{j-1},x_j)\Big) f(x_{n-1},x_n)g(x_n,y_n)^{\phi_{k,n}} 
\end{align*}
starting at $\tilde{\pi}_n^1$ and ending at $\tilde{\pi}_n^{K_n}=\pi_n$ for some $K_n\in\mathbb{N}$. Compute the weights
\begin{align*}
W_{1,n}^{(i)} = \bar{W}_{n-1}^{(i)}\cdot \frac{\ \tilde{\gamma}_n^1(x^{(i)}_{1:n-1},\tilde{x}_n^{(i)}) }{\eta_n(x^{(i)}_{1:n-1}, \tilde{x}_n^{(i)} )}  = \bar{W}_{n-1}^{(i)} \cdot  g(\tilde{x}_n^{(i)},y_n )^{\phi_{1,n}},
\end{align*}
normalize $\bar{W}_{1,n}^{(i)} = W_{1,n}^{(i)}/\sum_{j=1}^N W_{1,n}^{(j)}$.  
If $ESS \leq N^*$, resample $\{\tilde{x}_{1,n}^{(i)}\}$, set $\bar{W}_{1,n}^{(i)}=\frac{1}{N}$. \\ Set $k=1$ and denote the particles by $\{x_{k,n}^{(i)}\}$ and the weights $\{\bar{W}_{k,n}^{(i)}\}$. Set $\textit{flag}=0$. \\
Perform the following SMC sampler:

While $\textit{flag}=0$, do:
\begin{enumerate}
\item Find $\delta \in [0,1-\phi_{k,n}]$ so that $ESS(\delta)=N^*$. If $\delta=1-\phi_{k,n}$, set $\phi_{k+1,n}=1$, $K_n=k+1$ and $\textit{flag}=1$, else set $\phi_{k+1,n}=\phi_{k,n} +\delta$. 

\item Compute the weights 
\begin{align*}
W_{k+1,n}^{(i)} = \bar{W}_{k,n}^{(i)} \cdot \frac{ \tilde{\gamma}_n^{k+1}(x_{1:n-1}^{(i)}, x_{k,n}^{(i)} ) }{ \tilde{\gamma}_n^k(x_{1:n-1}^{(i)}, x_{k,n}^{(i)} )  } = \bar{W}_{k,n}^{(i)} \cdot g(x_{k,n}^{(i)},y_n)^{\delta},
\end{align*}
normalize $\bar{W}_{k+1,n}^{(i)} = W_{k+1,n}^{(i)}/\sum_{i=1}^N W_{k+1,n}^{(i)}$. If $ESS \leq N^*$, resample $\{x_{k,n}^{(i)}\}$, set $W_{k+1,n}^{(i)}=\frac{1}{N}$. Denote the particles by $\{\hat{x}_{k,n}^{(i)}\}$. Attach these to the samples $\{x_{1:n-1}^{(i)}\}$ from previous time steps and denote the complete paths $\{\hat{x}_{k,1:n}^{(i)}\}$. 

\item (MCMC) For $m\in\{1,\ldots,S\}$, do: 
\begin{itemize}
\item For $i\in\{1,\ldots, N\}$, sample $x_{k+1,1:n}^{(i)}|  \hat{x}_{k,1:n}^{(i)}$ from a Markov kernel that preserves
\[\tilde{\pi}_n^{k+1}(x_{1:n}) \propto 
\Big(\prod_{j=1}^{n-1}g(x_j,y_j)\,f(x_{j-1},x_j)\Big)
f(x_{n-1},x_n)\,g(x_n,y_n)^{\phi_{k+1,n}}.   \] 
Again, one can use a random walk: $x^{'(i)}_{1:n} = \hat{x}_{k,1:n}^{(i)} + Z^{(i)},\quad Z^{(i)} \sim \mathcal{N}(0,\Sigma_m)$, for some covariance matrix $\Sigma_m$, so that the acceptance probability is
\begin{align*}
\alpha_m^{(i)} &= 1 \wedge  \frac{ \tilde{\pi}_n^{k+1}(x^{'(i)}_{1:n})  } {  \tilde{\pi}_{n}^{k+1}(\hat{x}_{k,1:n}^{(i)}) }  = 1 \wedge  \frac{\prod_{j=1}^{n-1}g(x^{'(i)}_j,y_j)\,f(x^{'(i)}_{j-1}, x^{'(i)}_j) }{ \prod_{j=1}^{n-1}g(\hat{x}_{k,j}^{(i)},y_j)\,f(\hat{x}_{k,j-1}^{(i)},\hat{x}_{k,j}^{(i)})} \frac{  f(x^{'(i)}_{n-1}, x^{'(i)}_n)\, g(x^{'(i)}_{n},y_n)^{\phi_{k+1,n}}  }{   f(x_{n-1}^{(i)},\hat{x}_{k,n}^{(i)})\,g(\hat{x}_{k,n}^{(i)},y_n)^{\phi_{k+1,n}}   }.
\end{align*}
With probability $\alpha_m^{(i)}$, set $\hat{x}_{k,1:n}^{(i)}=x^{'(i)}_{1:n}$.
\end{itemize}
\item Set $k=k+1$ and $x_{k,n}^{(i)} = \hat{x}_{k-1,n}^{(i)}$, $i\in\{1,\ldots,N\}$.
\end{enumerate}

Calculate $ESS$ for the last weights $\{\bar{W}_{K_n,n}^{(i)}\}$. If $ESS\leq N^*$, resample $\{x_{0:n-1}^{(i)},x_{K_n,n}^{(i)}\}$  and set $\bar{W}_{K_n,n}^{(i)}=\frac{1}{N}$. \\ Set $x_{0:n}^{(i)} = \{x_{0:n-1}^{(i)},x_{K_n,n}^{(i)}\}$ and $\bar{W}_{n}^{(i)} = \bar{W}_{K_n,n}^{(i)}$, $i\in\{1,\ldots, N\}$.
\item \textbf{Iterations $n \geq L+1$:} 

We have $N$ paths $\{x_{1:n-1}^{(i)}\}$ with weights $\{\bar{W}_{n-1}^{(i)}\}$. We begin by sampling particles $\tilde{x}_n^{(i)}$ from proposal $\eta_n(x_{1:n}) = f(x_{n-1}^{(i)},x_n)\widehat{\pi}_{n-1}(x_{1:n-1})$ (recalling that $\widehat{\pi}_L = \pi_L$),  $i\in\{1,\ldots, N\}$, via \eqref{eq:signal}. Setting $\mu_0(x_1)=f(x_0,x_1)$, we will sample sequentially from the distributions
\begin{align*}
\widetilde{\pi}_{n}^k(x_{1:n})  &\propto \Big( \widehat{\pi}_{n-1}(x_{1:n-1}) f(x_{n-1},x_{n}) \Big)^{1-\phi_{k,n}} \widehat{\pi}_{n}(x_{1:n})^{\phi_{k,n}} \\
& =\left( \frac{\mu_{n-L}(x_{n-L+1})\,g(x_{n},y_{n})  }{f(x_{n-L},x_{n-L+1})} \right)^{\phi_{k,n}} \Big[\prod_{j=0}^{n-L-1} \mu_j(x_{j+1})  \prod_{j=n-L+1}^{n} f(x_{j-1},x_j) \prod_{j=1}^{n-1} g(x_j,y_j) 
\Big]
\end{align*}
starting at $\tilde{\pi}_n^1$ and ending at $\tilde{\pi}_n^{K_n}=\widehat{\pi}_n$ for some $K_n\in\mathbb{N}$. Compute the weights
\begin{align*}
W_{1,n}^{(i)} = \bar{W}_{n-1}^{(i)}\cdot \frac{\ \tilde{\gamma}_n^1(x^{(i)}_{1:n-1},\tilde{x}_n^{(i)}) }{\eta_n(x^{(i)}_{1:n-1}, \tilde{x}_n^{(i)} )}=\bar{W}_{n-1}^{(i)}\left[ \frac{g(x_n^{(i)},y_n) \mu_{n-L}(x_{n-L+1}^{(i)}) }{f(x_{n-L}^{(i)},x_{n-L+1}^{(i)})} \right]^{\phi_{1,n}},
\end{align*}
normalize $\bar{W}_{1,n}^{(i)} = W_{1,n}^{(i)}/\sum_{i=1}^N W_{1,n}^{(i)}$. If $ESS \leq N^*$, resample $\{\tilde{x}_n^{(i)}\}$ and set $\bar{W}_{1,n}^{(i)}=\frac{1}{N}$. \\ Set $k=1$ and denote the particles by $\{x_{k,n}^{(i)}\}$ and the weights $\{\bar{W}_{k,n}^{(i)}\}$. Set $\textit{flag}=0$. \\ Perform the following SMC sampler:

While $\textit{flag}=0$, do:
\begin{enumerate}
\item Find $\delta \in [0,1-\phi_{k,n}]$ so that $ESS(\delta)=N^*$. If $\delta=1-\phi_{k,n}$, set $\phi_{k+1,n}=1$, $K_n=k+1$ and $\textit{flag}=1$, else set $\phi_{k+1,n}=\phi_{k,n} +\delta$. 

\item Compute the weights 
\begin{align*}
W_{k+1,n}^{(i)} = \bar{W}_{k,n}^{(i)}\cdot \frac{ \tilde{\gamma}_n^{k+1}(x_{1:n-1}^{(i)}, x_{k,n}^{(i)} ) }{ \tilde{\gamma}_n^k(x_{1:n-1}^{(i)}, x_{k,n}^{(i)} )  } =\bar{W}_{k,n}^{(i)}\cdot  g(y_n|x_{k,n}^{(i)})^{\delta},
\end{align*}
normalize $\bar{W}_{k+1,n}^{(i)} = W_{k+1,n}^{(i)}/\sum_{i=1}^N W_{k+1,n}^{(i)}$. If $ESS\leq N^*$, resample $\{x_{k,n}^{(i)}\}$, set $\bar{W}_{k+1,n}^{(i)}=\frac{1}{N}$. Denote the resulting samples $\{\hat{x}_{k,n}^{(i)}\}$. Attach these to the  samples $x_{n-L:n-1}^{(i)}$ from the previous $L$ time steps and denote the new samples by $\{\hat{x}_{k,n-L:n}^{(i)}\}$. 
\item (MCMC) For $m\in\{1,\ldots,S\}$: 
\begin{itemize}
\item For $i\in\{1,\ldots, N\}$, sample  $x_{k+1,n-L:n}^{(i)}|  \hat{x}_{k,n-L:n}^{(i)}$ from a Markov kernel that preserves
\begin{align*}
\widetilde{\pi}_{k+1,n}(x_{1:n}) &= \left( \frac{g(x_{n},y_{n}) \mu_{n-L}(x_{n-L+1}) }{f(x_{n-L},x_{n-L+1})} \right)^{\phi_{k+1,n}}\times \\
& \qquad\qquad \Big[ \prod_{j=0}^{n-L-1} \mu_j(x_{j+1})   \prod_{j=n-L+1}^{n} f(x_{j-1},x_j)\prod_{j=1}^{n-1} g(x_j,y_j) 
\Big]
\end{align*}
Again, one can use a random walk
\begin{align*}
x^{'(i)}_{n-L:n} &= \hat{x}_{k,n-L:n}^{(i)} + Z^{(i)},\quad Z^{(i)} \sim \mathcal{N}(0,\Sigma_m)
\end{align*}
for some covariance matrix $\Sigma_m$, so that the acceptance probability is 
{\footnotesize
\begin{align*}
 a_m^{(i)} &= 1\wedge \frac{\widetilde{\pi}_{k+1,n}(x^{'(i)}_{n-L:n})}{\widetilde{\pi}_{n}^{k+1}(\hat{x}_{k,n-L:n}^{(i)})} \\
&= 1 \wedge  \Bigg[ \left( \frac{g(x^{'(i)}_n,y_{n}) \mu_{n-L}(x^{'(i)}_{n-L+1}) }{f(x^{'(i)}_{n-L},x^{'(i)}_{n-L+1})}  \frac{  f(\hat{x}_{k,n-L}^{(i)},\hat{x}_{k,n-L+1}^{(i)}) }{ g(\hat{x}_{k,n-L+1}^{(i)},y_{n}) \,\mu_{n-L}(\hat{x}_{k,n-L+1}^{(i)}) }  \right)^{\phi_{k+1,n}} \times \\
&\qquad\qquad\qquad \frac{ \mu_{n-L-1}(x^{'(i)}_{n-L}) \prod_{j=n-L}^{n-1} g(x^{'(i)}_j,y_j) \prod_{j=n-L+1}^{n} f(x^{'(i)}_{j-1},x^{'(i)}_{j}) } {  \mu_{n-L-1}(\hat{x}_{k,n-L}^{(i)}) \prod_{j=n-L}^{n-1} g(\hat{x}_{k,j}^{(i)},y_j) \prod_{j=n-L+1}^{n} f(\hat{x}_{k,j-1}^{(i)},\hat{x}_{k,j}^{(i)}) ` }\Bigg].
\end{align*}
}
With probability $\alpha_m^{(i)}$, set $\hat{x}_{k,n-L:n}^{(i)}=x^{'(i)}_{n-L:n}$.
\end{itemize}
\item Set $k=k+1$ and $x_{k,n}^{(i)} = \hat{x}_{k-1,n}^{(i)}$, $i\in\{1,\ldots, N\}$.
\end{enumerate}
If $ESS \leq N^*$, resample $\{x_{0:n-1}^{(i)},x_{K_n,n}^{(i)}\}$  and set $\bar{W}_{K_n,n}^{(i)}=\frac{1}{N}$. \\ Set $x_{0:n}^{(i)} = \{x_{0:n-1}^{(i)},x_{K_n,n}^{(i)}\}$ and $\bar{W}^{(i)}_{n}=\bar{W}_{K_n,n}^{(i)}$, $i\in\{1,\ldots, N\}$.
\end{itemize}
\hrulefill
\end{center}

\end{document}